\documentclass[10pt]{article}
\usepackage[a4paper, margin=1in]{geometry}
\usepackage{hyperref}
\hypersetup{
    colorlinks=true,
    linkcolor=black,
    filecolor=magenta,      
    urlcolor=blue,
    citecolor=black,
}
\usepackage[numbers]{natbib} 
\usepackage{graphicx}
\usepackage{xcolor} 
\usepackage{amsmath,amssymb,mathtools} 
\usepackage{bm}   
\usepackage{color,soul}
\usepackage{amsmath,amssymb,mathtools}   
\usepackage{amsthm}                      
\usepackage{arydshln}
\usepackage{stmaryrd}
\usepackage{thm-restate}
\usepackage{booktabs}
\usepackage{multirow}
\usepackage{float}

\theoremstyle{plain}
\newtheorem{lemma}{Lemma}[section]

\theoremstyle{definition}
\newtheorem{definition}{Definition}[section]

\theoremstyle{remark}
\newtheorem*{remark}{Remark}

\usepackage[mathlines]{lineno}
\newcommand{\jz}[1]{\textcolor{black}{#1}}

\newcommand{\Eq}[1]{Eq.~(\ref{#1})}

\def \FP {\operatorname{FP}}
\def \FPe {\operatorname{FP_e}}

\def \supp {\operatorname{supp}\,}
\def \esupp {\operatorname{esupp}\,}

\def \EIFP {\FP(W,\theta,\tau_I)}
\def \EIFPe {\FPe(W,\theta,\tau_I)}

\def \EITLN {\operatorname{E-I\ TLN}}

\def \spec {\operatorname{spec}}

\def\od{\stackrel{\mathrm{def}}{=}}

\begin{document}

\begin{center}
{\Large{
Sequential chaotic oscillations \\ in excitatory--inhibitory threshold--linear networks
}} \\
\vspace{3mm}
{\large Jie Zang, Carina Curto} \\
\vspace{1.5mm}
{\normalsize Division of Applied Mathematics, Brown University} \\
{\normalsize Robert J. and Nancy D. Carney Institute for Brain Science, Brown University} \\
\vspace{1mm}
{\normalsize \texttt{jie\_zang@brown.edu}, \texttt{carina\_curto@brown.edu}} \\
\vspace{3mm}
Oct 2, 2026
\end{center}

Metastable states, a phenomenon observed in brain dynamics and many other systems, have been proposed as a key feature of healthy brain function, reflecting a balance between integration and segregation. However, it remains unclear how to capture this behavior within a dynamical-systems framework. In this paper, we propose sequential chaotic oscillations (SCOs), arising in excitatory-inhibitory threshold-linear networks (E-I TLNs), as a candidate dynamical mechanism for sequential metastability. As a simple form of chaotic itinerancy, SCOs occur under constant input and consist of a sequence of metastable states whose transition order can be predicted by the underlying graph. To identify the parameter regime for SCOs, we develop new graph rules for E-I TLNs and use them to characterize the fixed point structure of E-I TLNs on paths and cycles. Our results show that the emergence of SCOs requires unstable singleton fixed points and sufficiently strong inhibition. 
In addition to SCOs, we find that E-I oscillations need not be synchronized. Motivated by this, we introduce a decomposition into the $z$-mode and the mean mode, which capture excitatory differences and overall network activity, respectively. These modes are then used to distinguish attractors associated with the full-support fixed point of E-I TLNs on cycles.

\vspace{2mm}

\tableofcontents

\section{Introduction}
\label{sec:introduction}
Metastable states, characterized by long transient dynamics before transitioning to another state, have been observed in cortical neural activity~\cite{rabinovich2008transient, jones2007natural, mazzucato2015dynamics} as well as in many other systems, including thermodynamic processes~\cite{sun2016thermodynamic}, physical systems~\cite{franccoise2006encyclopedia}, and human motor coordination~\cite{kelso1984phase}. 
For example, in the gustatory cortex of rats, neural populations generate taste-specific sequences of metastable states~\cite{jones2007natural}. Such dynamics are thought to reflect a delicate balance between cooperation and relative independence among neural populations~\cite{kelso1995dynamic}. Consequently, sequential metastability has been emphasized as a fundamental feature of brain function~\cite{hancock2025metastability}. However, despite its broad relevance, metastability remains challenging to capture within the framework of deterministic dynamical systems.

Metastable states in dynamical systems are somewhat counterintuitive: trajectories can remain near such states for long periods of time, but eventually leave. As a result, they are not true attractors. One classical mechanism for generating such dynamics is the heteroclinic channel, a narrow region of phase space organized by an appropriate itinerary of heteroclinic orbits~\cite{rabinovich2010heteroclinic}. Such a channel consists of a sequence of saddle equilibria, causing the system to successively approach each saddle, remain near it for some time, and then move on to the next~\cite{rabinovich2020sequential}. However, heteroclinic channels typically arise only in specially structured models.

Another possibility is chaotic itinerancy, a dynamical phenomenon often associated with high dimensional systems~\cite{kaneko2003chaotic}. It was independently identified in models of optical turbulence~\cite{ikeda1989maxwell}, globally coupled chaotic systems~\cite{kaneko1990clustering,kaneko1991globally}, and nonequilibrium neural networks~\cite{tsuda1991chaotic,tsuda1992dynamic}. Chaotic itinerancy has been described as itinerant (wandering) motion among multiple low-dimensional states, mediated by high-dimensional chaos~\cite{ikeda1989maxwell}. 
These low-dimensional states exhibit transient stability even though they are not true attractors. For this reason, they are referred to as attractor ruins~\cite{kaneko2003chaotic}. Chaotic itinerancy is analogous to a heteroclinic channel, except that trajectories pass near a set of attractor ruins rather than a set of saddle fixed points. In many known examples~\cite{ikeda1989maxwell,tsuda1991chaotic,tsuda1992dynamic}, such behavior relies on sufficiently high-dimensional dynamics and tends to produce irregular, difficult-to-control sequences of attractor ruins.

\begin{figure}
    \centering
    \includegraphics[width=1\linewidth]{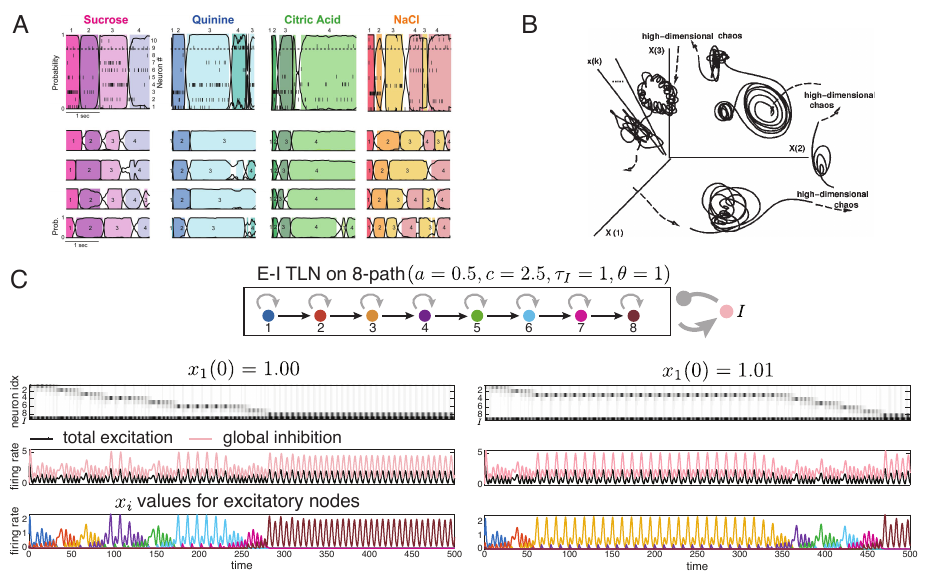}
    \caption{Motivation of our work.
    (A) Taste-specific sequences of metastable states in the rat gustatory cortex, reproduced from Fig.~2 of~\cite{jones2007natural}.
    (B) Schematic representation of chaotic itinerancy, reproduced from Fig.~1 of~\cite{kaneko2003chaotic}.
    (C) Sequential chaotic oscillations in the $8$-path E-I TLN. The initial condition $x_1(0)$ is indicated in each simulation, with $x_2(0)=\cdots=x_8(0)=0.01$ and $x_I(0)=0.1$ for both cases.}
    \label{fig:background}
\end{figure}

Here we show that chaotic itinerancy can arise in a simple class of piecewise-linear recurrent neural networks, namely excitatory-inhibitory threshold-linear networks (E-I TLNs). These networks are generated from directed graphs, and the organization of their linear subsystems is encoded by the underlying graph structure. In this setting, the sequence of metastable states, or attractor ruins in the sense of~\cite{kaneko2003chaotic}, is also graph-determined. 
\jz{We refer to this behavior as sequential chaotic oscillations (SCOs) because the overall switching dynamics appear chaotic and each individual attractor ruin also exhibits chaotic dynamics; however, the sequence of visits retains the same order irrespective of initial conditions (see Figure~\ref{fig:background}). }
In contrast to many known examples~\cite{ikeda1989maxwell,kaneko1990clustering,kaneko1991globally,tsuda1991chaotic,tsuda1992dynamic}, this mechanism can already arise in low-dimensional systems, such as the three-dimensional E-I TLN on a $2$-path. 

To identify the parameter conditions that support  SCOs, it is necessary to analyze the fixed point structure of E-I TLNs, since this provides the dynamical foundation for the emergence of SCOs. This analysis relies on the graphical domination theorem inherited from previous work~\cite{curto2025graphical} as well as new graph rules developed here for E-I TLNs. Using these rules, we characterize the complete fixed point structure of nondegenerate E-I TLNs, focusing on two graph families: the $n$-node path and the $n$-node cycle (Theorems~\ref{thm:path_supports} and \ref{thm:cycle_supports}). In both cases, the emergence of SCOs requires unstable singleton fixed points together with sufficiently strong inhibition.

In E-I TLNs, SCOs arise from a hybrid mechanism combining E-I oscillations with the underlying graph structure. The E-I oscillatory component produces prolonged chaotic activation of each node, while the graph structure determines the firing order of the excitatory nodes. The results for paths suggest that chaotic itinerancy need not be driven by high-dimensional chaos; rather, it may itself serve as a building block of high-dimensional chaotic dynamics, as illustrated by E-I TLNs on cycles. We emphasize that the attractors and chaotic behaviors reported here are identified numerically using MATLAB's \texttt{ode45} solver and are not established by rigorous mathematical proof. Nevertheless, because the system is continuous and piecewise linear, we expect the numerical simulations to be reliable.

Besides the singleton fixed points associated with SCOs, the full-support fixed point of an E-I TLN on a cycle is also interesting since it can give rise to at least three distinct attractors when it becomes unstable. To distinguish parameter regimes for these attractors, we introduce a new $z$-mode and mean mode decomposition. 
This decomposition provides a new perspective by separately capturing activity differences and total activity, thereby going beyond a stability-based analysis of fixed points alone (Theorem~\ref{thm:cycle_2_modes}). It allows us to classify the parameter regimes associated with different attractors, including E-I oscillations, CTLN-like oscillations, and flower-like attractors.

Finally, the paper is organized as follows. Section~\ref{sec:introduction} introduces the background and motivation for our work. Section~\ref{sec:preliminaries} introduces the model and basic definitions. Section~\ref{sec:results} presents the main results on fixed point structure and attractor classification, including Theorems~\ref{thm:path_supports}, \ref{thm:cycle_supports}, and \ref{thm:cycle_2_modes}. Detailed proofs are provided in the Appendix. Section~\ref{sec:discussion} discusses the implications of these results.

\section{Preliminaries}
\label{sec:preliminaries}
Here we briefly recall the E-I TLN framework and its relation to CTLNs. E-I TLNs are graph-based networks introduced in~\cite{curto2025graphical} as an extension  of combinatorial threshold-linear networks (CTLNs)~\cite{morrison2024diversity}, which were designed to investigate the link between graph structure and network dynamics~\cite{curto2024stable,parmelee2022sequential}. CTLNs consist entirely of excitatory nodes, while inhibition is incorporated implicitly through a global ``blanket of inhibition''~\cite{fino2011dense} that enforces competition. The development of E-I TLNs was motivated in part by the work of Lienkaemper et al.~\cite{lienkaemper2025diverse}, who showed that E-I spiking networks can be reduced to CTLNs by tuning the inhibitory timescale. Building on this insight, the E-I TLN was introduced in~\cite{curto2025graphical}, bringing  an explicit inhibitory node back into a graph-based threshold-linear network.

Due to the similarity between E-I TLNs and CTLNs, it was shown in~\cite{curto2025graphical} that for any positive inhibitory timescale, the fixed points of E-I TLNs can be matched with those of the corresponding CTLN. Thus, existing CTLN theory~\cite{parmelee2022core,curto2023graph}, including domination theory~\cite{curto2025graphical}, applies to E-I TLNs, but only within the moderate inhibition regime. In particular, the strong and weak inhibition regimes fall outside the scope of previously known graph rules. To address this gap, we develop a new set of parameter-dependent graph rules that either guarantee or exclude specific subsets as fixed point excitatory supports. 

\paragraph{E-I TLN dynamics.}
 An E-I TLN is a specific TLN model consisting of $n$ excitatory nodes together with a single global inhibitory node $I$. The dynamics of the excitatory nodes, $x_1,\ldots, x_n$, and the inhibitory node, $x_I$, follow the standard threshold-linear form:
 \begin{subequations}\label{eq:EI_TLN}
 \begin{align}
 \tau_E \frac{dx_i}{dt} &= -x_i + \left[ \sum_{j=1}^{n} W_{ij}x_j + W_{iI}x_I + b_i \right]_+ \,\,\, \od -x_i + [y_i]_+, \quad i = 1, \dots, n, \\
 \tau_I \frac{dx_I}{dt} &= -x_I + \left[ \sum_{j=1}^{n} W_{Ij}x_j + W_{II}x_I + b_I\right]_+\od -x_I + [y_I]_+.
 \end{align}
 \end{subequations}
 Here, $\tau_E, \tau_I > 0$ denote the timescales of the excitatory and inhibitory nodes, respectively. We will measure time in units of $\tau_E$, and thus set $\tau_E = 1$. The function $[y]_+ = \max\{y,0\}$ denotes the standard ReLU activation, and $y_i$ and $y_I$ denote the terms inside the threshold. The external input is given by $b = (b_1,\dots,b_n,b_I)^T$.
 
What makes this an E-I TLN (as opposed to just a TLN) is the designation of nodes $[n] \od \{1,\ldots,n\}$ as excitatory, and node $n+1$ as inhibitory, which we will denote with index label $I$. Accordingly, the $W$ matrix satisfies Dale's law, with
$$ W_{ij} \geq 0, \quad W_{Ij} > 0, \quad W_{iI} \leq 0,  \;\; \text{and} \;\; W_{II} \leq 0 \text{ for } i,j \in [n], \; I = n+1.$$

\paragraph{Fixed points of an E-I TLN.} 
For such an E-I TLN, a \emph{fixed point} is defined, as usual, to be a point $x^* \in \mathbb{R}^{n+1}$ such that
 \begin{equation*}
 \left.\frac{dx_i}{dt}\right|_{x=x^*}=0, \quad \text{ for all } i \in [n]\cup\{I\}.
 \end{equation*}
For simplicity, and in anticipation of our graph-based E--I TLNs, we assume here that $b_I = 0$, $b_i = \theta > 0$ for all $i \in [n]$, and $W_{II} = 0$. Under these assumptions, it is straightforward to check that $x_I^* > 0$ for every fixed point. This observation motivates us to define, in addition to the \emph{support} of $x^*$ as the set of active nodes, the \emph{e-support} of $x^*$ as the set of active excitatory nodes:
 \begin{equation*}\label{eq:support}
 \supp (x^*) \; \od \; \{\, i \in [n]\cup\{I\} \mid x_i^* > 0 \,\}, \qquad 
 \esupp (x^*) \; \od \; \{\, i \in [n] \mid x_i^* > 0 \,\}.
 \end{equation*}
Since $x_I^*>0$ for every fixed point, we have $\supp(x^*) = \esupp(x^*) \cup \{I\}$.

In nondegenerate TLNs, there is at most one fixed point with a given support~\cite{curto2019fixed}. Similarly, for nondegenerate E-I TLNs, we find that there is at most one fixed point with a given e-support (see the full statement and proof in the Appendix). Therefore, the collection of all supports or e-supports of an E-I TLN is in one-to-one correspondence with the set of all fixed points. We denote these collections by $\EIFP$ and $\EIFPe$, respectively:
\begin{align}
\EIFP
&\od
\bigl\{\, \sigma \cup \{I\} \mid
\sigma \cup \{I\} = \supp(x^*) \text{ for some fixed point } x^*
\text{ of the } \EITLN \bigr\},
\label{eq:FP}
\\
\EIFPe
&\od
\bigl\{\, \sigma \mid
\sigma = \esupp(x^*) \text{ for some fixed point } x^*
\text{ of the } \EITLN \bigr\}.
\label{eq:FPe}
\end{align}
Note that for any $\sigma \in \EIFPe$, the value of the corresponding fixed point $x^*$ is easily recovered: $x_k^* = 0$ for $k \in [n]\setminus \sigma$, and the positive values $x_\sigma^*$ and $x_I^*$ are given by:
$$ x^*_{\sigma \cup \{I\}} = (I-W_{\sigma \cup \{I\}})^{-1}b_{\sigma \cup \{I\}}.$$
Here, $x^*_{\sigma \cup \{I\}}$ and $b_{\sigma \cup \{I\}}$ denote the column vectors obtained by restricting to the indices in $\sigma \cup \{I\}$, and $W_{\sigma \cup \{I\}}$ denotes the principal submatrix of $W$ obtained by restricting both rows and columns to $\sigma \cup \{I\}$.

\begin{figure}
    \centering
    \includegraphics[width=0.9\linewidth]{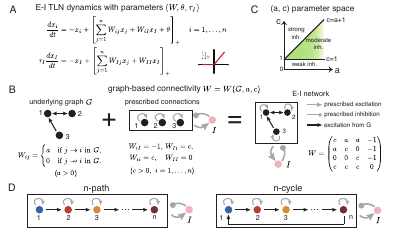}
    \caption{
    Construction of an E-I TLN.
    (A) Graph-based E-I TLN equations.
    (B) Example illustrating how the connectivity matrix $W$ is constructed from a graph $G$.
    In the connectivity matrix $W$, the inhibitory neuron is indexed as node $n+1$, so the last row and last column correspond to the prescribed weights $W_{Ii}$ and $W_{iI}$, respectively.
    (C) Three inhibition regimes in the $(a,c)$ parameter space, separated by the lines $c=1$ and $c=a+1$.
    (D) The two graph families studied in this paper.
}
    \label{fig:EI_W_setting}
\end{figure}

\paragraph{Graph-based E-I TLNs.} In this work, we will focus on graph-based E-I TLNs, where the connectivity matrix $W$ comes from a directed graph. The graph $G$ specifies the interactions among excitatory nodes with each vertex corresponding to one excitatory node. For $i,j = 1, \dots, n$, the connections defined from $G$ are given by
 \begin{align*}\label{eq:EI_W}
    W_{ij} = 
     \begin{cases}
         a & \text{if } j \to i \text{ in } G, \\
         0   & \text{if } j \nrightarrow i \text{ in } G,
     \end{cases}
 \end{align*}
 where $a>0$ represents the excitation weight.  All remaining (non-graph-based) connections are set by prescribed values. For any excitatory node $i$ and the inhibitory node $I$,
 \begin{equation*}
     W_{ii} = c,\, W_{iI} = -1,\, W_{Ii} = c, \,\text{and} \,\, W_{II} = 0,\quad i = 1,\dots,n,
 \end{equation*}
 where $c>0$ represents the inhibition contribution from excitatory nodes. Unlike other TLNs (such as CTLNs), each excitatory neuron $i$ in E-I TLNs includes a self-excitation term~\cite{curto2025graphical} whose magnitude is set to $W_{ii}=-W_{iI} W_{Ii}$. This condition ensures a balance between inhibitory and excitatory interactions. A construction of the connectivity matrix $W$ for an E–I TLN is illustrated in Figure~\ref{fig:EI_W_setting}B. In addition, to focus on internally-generated (rather than input-driven) dynamics, we take $b=(\theta,\dots,\theta,0)^T$,
where $\theta>0$ is the input to each excitatory node and the inhibitory node receives no external input. This is without loss of generality, since any constant input to the inhibitory node can be absorbed by a change of variables, yielding an equivalent system with $b_I=0$.

 Following the above definition, an E-I TLN is fully specified by its underlying directed graph $G$, the excitation weight $a$, the inhibition weight $c$, the external input $\theta$, and the inhibitory timescale $\tau_I$. We therefore denote such a network by $(G,a,c, \theta, \tau_I)$. 
 
Throughout this paper, we focus on the positive $(a,c)$ parameter space, which is divided into the three inhibition regimes shown in Figure~\ref{fig:EI_W_setting}C:
\[
\begin{aligned}
\text{strong inhibition:} \qquad & a>0,\quad c>a+1,\\
\text{moderate inhibition:} \qquad & a>0,\quad 1<c<a+1,\\
\text{weak inhibition:} \qquad & a>0,\quad 0<c<1.
\end{aligned}
\]
 The moderate inhibition regime corresponds to the parameter range for CTLNs, where a series of graph rules is available~\cite{curto2024stable,curto2023graph}. A major difference in this inhibition regime between CTLNs and E-I TLNs is that they have different fixed-point stability. For example, singleton fixed points are always stable in CTLNs, whereas in E-I TLNs they may become unstable. By contrast, the other inhibition regimes require new graph rules. In addition, we focus on two graph families, namely $n$-paths and $n$-cycles (see Figure~\ref{fig:EI_W_setting}D), defined as follows:
\begin{definition}[$n$-path and $n$-cycle]
The \emph{$n$-path} ($n\ge 1$) is the directed graph on vertex set $[n] = \{1,\dots,n\}$ with edges $i\to i+1$ for $i=1,\dots,n-1$.  
The \emph{$n$-cycle} ($n\ge 3$) is the directed graph on vertex set $[n]$ with edges $i\to i+1$ for $i\in [n]$, where indices are taken modulo $n$.
\end{definition}
Note that for an $n$-path, $n \ge 1$ is sufficient to form a path. Thus, a single node with no edges, i.e., a singleton, is a $1$-path. Note that these definitions refer only to the underlying excitatory graph; the full E-I TLN includes an additional inhibitory node and more edges.

\section{Results}
\label{sec:results}
Here we show that the emergence of sequential chaotic oscillations requires unstable singleton fixed points (supported on one excitatory node together with the inhibitory node) and sufficiently strong inhibition $(c>a+1)$ for E-I TLNs on paths and cycles. To establish this, we first analyze the fixed-point structure of nondegenerate E-I TLNs. Nondegeneracy ensures that each support corresponds to at most one fixed point, and it holds for generic parameter choices. The definition of nondegeneracy and the proof of fixed-point uniqueness are given in Section~\ref{sec:link_CTLN}.

\begin{figure}[t]
    \centering
    \includegraphics[width=1\linewidth]{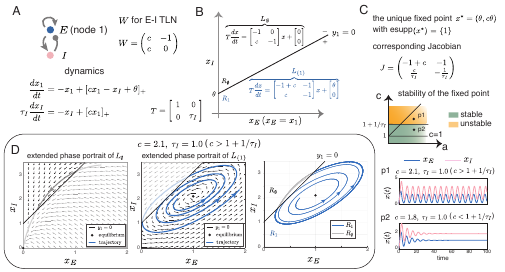}
    \caption{
    E-I oscillations in an E-I TLN whose underlying graph is a singleton.
    (A) The circuit diagram, connectivity matrix, and dynamics of the singleton E-I TLN.
    (B) The original E-I TLN decomposes into two linear subsystems, $L_{\emptyset}$ and $L_{\{1\}}$, separated by the  hyperplane $y_{1} = c x_{1} - x_{I} + \theta = 0$. $T=(1,0;0,\tau_I)$ is the timescale matrix.
    (C) The unique fixed point supported on $\{1,I\}$ and the associated dynamical behaviors (stable equilibrium and E-I oscillation), depending on the stability of this fixed point.
    (D) Phase portraits of $L_{\emptyset}$ and $L_{\{1\}}$ with $c=2.1$ and $\tau_I=1$. The gray and blue trajectories illustrate the influence of the stable node in $L_\emptyset$ (left) and the unstable focus in $L_{\{1\}}$ (middle). The right figure shows how switching between the two subsystems produces a stable limit cycle that alternates between the chambers $R_{1}$ (governed by $L_{\{1\}}$) and $R_{\emptyset}$ (governed by $L_{\emptyset}$). The external input is set to $\theta =1$.
     }
    \label{fig:singleton}
\end{figure}

\subsection{Singleton E-I TLN}
We begin with the most elementary instance of an E-I TLN: the singleton E-I TLN. We call it ``singleton'' because its underlying graph $G$ is a singleton, even though the full E-I TLN contains two nodes: one excitatory node and one inhibitory node.
Despite its simplicity, this case already illustrates how excitation, inhibition, and the ReLU nonlinearity interact to generate E-I oscillations, and it serves as a useful baseline for understanding more complex networks.

In this setting, the E-I TLN dynamics are given by ($\tau_E=1$)
\begin{align*}
    \frac{dx_1}{dt} &= -x_1 + \left[  c x_1 - x_I + \theta \right]_+,  \\
    \tau_I \frac{dx_I}{dt} &= -x_I + \left[ cx_1 \right]_+.
\end{align*}
Following the definition of the $y$ variables, we have $y_1=c x_1 -x_I+\theta$ and $y_I=c x_1$. For E-I TLNs, the dynamics are bounded below, i.e. $x\ge 0$ (see Section~\ref{sec:link_CTLN} for the proof). Hence, we restrict attention to the region $x_1 \geq 0$, $x_I \geq 0$, where the boundary $y_1(x)=0$ divides the phase space into two chambers:
\[
R_{\emptyset}=\{x \mid y_1(x)\leq 0,\; y_I(x)\geq 0\}, \qquad
R_{\{1\}}=\{x \mid y_1(x)>0,\; y_I(x)\geq 0\}.
\]
Correspondingly, the system is divided into two linear subsystems: $L_{\emptyset}$, which governs the dynamics in $R_{\emptyset}$, and $L_{\{1\}}$, which governs the dynamics in $R_{\{1\}}$. Throughout the paper and in the figures, we may write $R_1$ for $R_{\{1\}}$, 
$R_{12}$ for $R_{\{1,2\}}$, and so on, for simplicity.

Since $y_I = c x_1 \geq 0$ and $y_1\le 0$, the linear system $L_\emptyset$ is given by
\begin{align*}
     \frac{dx_1}{dt} &= -x_1,  \\
     \tau_I \frac{dx_I}{dt} &= -x_I +  cx_1.
\end{align*}
The fixed point for $L_\emptyset$ is $(x^*_1,x^*_I) = (0,0)$ with eigenvalues $-1$ and $-1/\tau_I$, so it is a stable node. However, at this point we have $y_1(x^*)=\theta>0$, meaning that $(0,0)$ does not lie in the chamber $R_\emptyset$ where $L_\emptyset$ is valid (see Figure~\ref{fig:singleton}D). \jz{Such a fixed
point is referred to as a \emph{virtual fixed point (virtual equilibrium)} in the piecewise-smooth
dynamical systems literature~\cite{bernardo2008piecewise}.}
Therefore the fixed point for $L_\emptyset$ cannot serve as a fixed point of the original E-I TLN. 

The other linear subsystem is $L_{\{1\}}$, which is given by 
\begin{align*}
    \frac{dx_1}{dt} &= -x_1 +  c x_1 - x_I + \theta ,  \\
    \tau_I \frac{dx_I}{dt} &= -x_I +  cx_1.
\end{align*}
 $L_{\{1\}}$ has a fixed point $(x^*_1,x^*_I) = (\theta,c\theta)$. Since $y_1(x^*)=\theta>0$, this fixed point exists in the chamber $R_1$, which is governed by $L_{\{1\}}$. Hence it survives as a fixed point of the full system with the support $\{1,I\}$ and e-support $\{1\}$. Moreover, by computing the Jacobian matrix
\begin{equation*}
J =
\begin{pmatrix}
1 & 0\\
0 & \frac{1}{\tau_I}
\end{pmatrix}
\begin{pmatrix}
c-1 & -1\\
c & -1
\end{pmatrix}
=
\begin{pmatrix}
    c - 1 & -1 \\
    \frac{c}{\tau_I} & -\frac{1}{\tau_I}
\end{pmatrix},
\end{equation*}
the eigenvalues shows that the fixed point is stable for $c < 1 + 1/\tau_I$ and unstable for $c > 1 + 1/\tau_I$.
By a Poincaré-Bendixson argument, once the unique fixed point becomes unstable and trajectories remain bounded, the system must admit a periodic orbit. \jz{This transition is reminiscent of a Hopf bifurcation since a complex-conjugate pair of eigenvalues crosses the imaginary axis at $c = 1 + 1/\tau_I$ and a limit cycle emerges. However, it is not a classical Hopf bifurcation. In the classical Hopf
bifurcation theorem, the limit cycle emerges close to the fixed point and requires nonlinearity in a neighborhood of the fixed point (see Theorem~1 in Section~4.4 of \cite{perko2013differential}). Here, by contrast, the dynamics are exactly linear in a neighborhood of the fixed point, since it lies in the interior of a single chamber. Therefore, no limit cycle can be generated locally around the fixed point; instead, the periodic orbit arises at some distance from the fixed point and must intersect more than one linear chamber.}

In this E-I system, this limit cycle arises from the combination of the trajectories emerging from the unstable focus in $L_{\{1\}}$ and those converging to the stable node in $L_\emptyset$ (which does not survive in the E–I TLN), as illustrated in Figure~\ref{fig:singleton}C–D. The periodic solution (E–I oscillation) reappears later in more complex E–I TLNs as a collective phenomenon and serves as a precursor to richer dynamical behaviors, including chaotic and quasi-periodic dynamics.

\subsection{Graph rules for E-I TLNs}
\jz{In studies of TLNs, the fixed point structure is often the natural starting point for understanding the dynamics~\cite{parmelee2022core,curto2023graph}, as in the singleton E-I TLN. The stability of these fixed points can then be used to identify parameter regimes in which attractors may arise. Accordingly, for E-I TLNs on $n$-paths and $n$-cycles, we first characterize the fixed point structure and then use fixed point stability to guide our study of the resulting attractors. Our analysis builds on the graphical domination theorem for CTLNs~\cite{curto2025graphical} together with several new graph rules, which are developed here for nondegenerate E-I TLNs. These rules are similar in spirit to graph rules for CTLNs~\cite{curto2023graph}, but also apply outside the CTLN range, i.e., the moderate inhibition regime.}

 Before stating the graph rules for fixed point structure, we make a simple but important observation: for any E-I TLN, the collections of fixed point supports and e-supports, \jz{denoted by \jz{$\FP$ and $\FPe$} in \Eq{eq:FP} and \Eq{eq:FPe}}, are independent of $\tau_I$ and $\theta$ (provided they are positive). This explains why the statements of all graph rules depend only on the underlying graph $G$ and the parameters $a,c$.
 In particular,
 \begin{equation}
     \EIFPe = \FPe(W) = \FPe(G,a,c),
 \end{equation}
 where the last equality follows from the fact that the connectivity matrix $W$ is determined by the underlying graph $G$ and parameters $a,c$. The proof is given in the Appendix. 

We begin by recalling the graphical domination theorem. As shown in~\cite{curto2025graphical} and in the Appendix, CTLNs and E-I TLNs have equivalent fixed point structures. Therefore, the graphical domination theorem for CTLNs applies directly to E-I TLNs in the moderate inhibition regime, as do other graph rules developed for CTLNs. We recall the definition below.

\begin{definition}[Graphical domination]
\label{def:domination1}
    Let $j,k \in [n]$ be vertices of $G$. We say that $k$ \emph{graphically dominates} $j$ in $G$ if the following two conditions hold:
    \begin{enumerate}
    \item For each $i \in [n]\setminus\{j,k\}$, if $i \to j$ then $i \to k$.
    \item $j \to k$ and $k \not\to j$.
    \end{enumerate}
\end{definition}
If such a pair $(j,k)$ exists, then $j$ is called a dominated node (or dominated vertex) of $G$.  For both CTLNs and E-I TLNs, it imposes strong combinatorial constraints on possible fixed point supports; see Lemma~3.8 of~\cite{curto2025graphical}. Here we restate it as Rule 1 in the moderate inhibition regime of E–I TLNs.

Beyond the moderate inhibition regime, we identify another graphical relation that constrains the fixed point structure in the weak inhibition regime. We call this relation \emph{weak graphical domination}, obtained by dropping the requirements $j\to k$ and $k\not\to j$ from Definition~\ref{def:domination1}.

\begin{definition}[Weak graphical domination]
    Let $j,k \in [n]$ be vertices of $G$. We say that $k$ \emph{weakly graphically dominates} $j$ in $G$ if for each $i \in [n]\setminus\{j,k\}$, $i \to j $ implies $i \to k$.
\end{definition}

Unlike graphical domination, weak graphical  domination imposes no condition on the edge between $j$ and $k$, and may therefore be symmetric. In particular, if $j$ and $k$ receive the same incoming edges from the rest of $G$, then they weakly graphically dominate each other. For example, in the $n$-path $1\to\cdots\to n$, node $i$ weakly dominates node $i+1$ for each $i\in[n-1]$, and node $1$ is weakly dominated by every other node. The resulting constraints on fixed point supports of nondegenerate E-I TLNs in the moderate ($1<c<a+1$) and weak ($0<c<1$) inhibition regimes are summarized in Rules~1 and~2 below; complete statements (Lemmas~\ref{lem:domination} and~\ref{lem:weakdomination}) and proofs are given in the Appendix.

The following six graph rules apply to any directed graph $G$. In the next section, however, we focus on specific graph families, such as paths and cycles, together with their induced subgraphs $G|_\sigma$. Recall that $G|_\sigma$ denotes the induced subgraph of $G$ obtained by restricting the vertex set to $\sigma$ and retaining all edges of $G$ between vertices in $\sigma$.

\medskip
 \noindent \textbf{Rule 1 }(domination reduction). \label{rule:domination}
If $j$ is a dominated node in the graph $G$, then in the moderate inhibition regime, 
\(
\FPe(G,a,c)
=
\FPe\!\left(G|_{[n]\setminus\{j\}},a,c\right)
\).

\medskip
\noindent \textbf{Rule 2 }(weak domination recruitment). \label{rule:weakdomination}
If $k$ weakly dominates $j$ in $G$, then in the weak inhibition regime, for each $\sigma \in \FPe(G,a,c)$, $j \in \sigma \Rightarrow k \in \sigma$.
\medskip

Here, $G|_{[n]\setminus\{j\}}$ denotes the induced subgraph obtained by removing node $j$ from $G$.
Although Rules~1 and~2 strongly constrain the possible fixed point supports, they are not sufficient to determine the complete fixed point structure. For this, we examine the \emph{on-neuron} and \emph{off-neuron} conditions for each candidate support.

Let $x^*$ be a fixed point of an E-I TLN. A subset $\sigma\cup\{I\}$ is its support if $x_i^*>0$ for $i\in\sigma\cup\{I\}$ and $x_j^*=0$ for $j\in[n]\setminus\sigma$. Accordingly, the fixed point conditions naturally split into on- and off-neuron conditions:
\begin{enumerate}
    \item  \emph{On-neuron condition:} for each $i \in \sigma \cup \{I\}$, $x^*_i = y_i(x^*)>0$.  This is exactly the condition for the restricted E-I TLN on $G|_\sigma$ to admit a full-support fixed point. Equivalently, $\sigma\in\FPe(G|_\sigma,a,c)$.
    \item \emph{Off-neuron condition:} for each $k\in [n]\setminus \sigma$, $y_k(x^*)\leq 0$. This ensures that neurons outside $\sigma$ remain silent, so a fixed point of the restricted E-I TLN on $G|_\sigma$ extends to the full network. Thus, if the off-neuron condition is satisfied, we have
    \[
    \sigma\in\FPe(G|_\sigma,a,c)
    \Rightarrow
    \sigma\in\FPe(G,a,c).
    \]
\end{enumerate}
Conversely, it is always true that $ \sigma\in\FPe(G,a,c)\Rightarrow \sigma\in\FPe(G|_\sigma,a,c)$ by the on-neuron condition.
The on- and off-neuron conditions together determine whether a subset $\sigma\subseteq[n]$ is a fixed point e-support. For the paths and cycles studied here, every proper induced subgraph is a disjoint union of directed paths. Thus, the analysis of the whole fixed point structure naturally reduces to understanding the on- and off-neuron conditions for such subgraphs. To facilitate this analysis, we establish several results for induced subgraphs, including rules for uniform in-degree subgraphs, unions of paths, and arbitrary subgraphs. 

Here, we say $G|_\sigma$ has \emph{uniform in-degree} $d$ if every node in $\sigma$ receives exactly $d$ incoming edges \emph{within} $G|_\sigma$. We summarize these results in the following graph rules for the three inhibition regimes with $a>0$: strong ($c>a+1$), moderate ($1<c<a+1$), and weak ($0<c<1$). Full statements (Lemmas~\ref{lem:uniform_indegree_on}-\ref{lem:off condition}) and proofs are provided in the Appendix.

\medskip
\noindent \textbf{Rule 3 }(uniform in-degree, on-neuron condition). 
If $G|_\sigma$ has uniform in-degree $d$, then
$\sigma \in \FPe(G|_\sigma,a,c)$ if and only if $(|\sigma|-1)c-da+1>0$.

\medskip
\noindent \textbf{Rule 4 }(uniform in-degree, off-neuron condition). 
If $G|_\sigma$ has uniform in-degree $d$ and
$\sigma\in\FPe(G|_\sigma,a,c)$, then $\sigma\in\FPe(G,a,c)
\iff
c\ge (d_k-d)a+1$ for all $k\in[n]\setminus\sigma$, where $d_k$ is the number of edges from $\sigma$ to $k$.

\medskip
\noindent\textbf{Rule 5 }(union of paths). 
If $G|_{\sigma}$ is a disjoint union of directed paths, then in the weak and strong inhibition regimes, $\sigma \in \FPe(G|_{\sigma},a,c)$.

\medskip
\noindent\textbf{Rule 6 }(guaranteed survival). 
Let $\sigma \in \FPe(G|_{\sigma},a,c)$ and $d_k$ be the number of edges from $\sigma$ to $k$. If $c \ge d_k a + 1$ for all $k\in[n]\setminus\sigma$, then $\sigma \in \FPe(G,a,c)$.
\medskip

Rule~6 indicates that with sufficiently strong inhibition (parameter $c$), any support in the restricted E-I TLN can survive as a fixed point of the full network. Note that for $n$-paths and $n$-cycles, $d_k \le 1$ for any subset $\sigma$ and any node $k \in [n]\setminus \sigma$. Thus, $c\ge a+1$ becomes the survival condition for many supports.

\subsection{E-I TLNs on $n$-path}

We now present the fixed point structure of E-I TLNs on paths in Theorem~\ref{thm:path_supports}. This theorem gives a complete breakdown of the parameter space and guides us in identifying the different dynamical behaviors, including sequential chaotic oscillations (SCOs). The proof relies on the graph rules developed above.
\begin{restatable}[Fixed points for paths]{theorem}{thmone}\label{thm:path_supports}
    Let $G$ be the $n$-path ($1\to2\to\dots \to n$). Then, for any nondegenerate E-I TLN with $a>0$ and $c>0$, the collection of fixed point e-supports $\FPe(G,a,c)$ satisfies:
    \begin{itemize}
        \item[\textnormal{(i)}]  If $c > a+1$ (strong inhibition), then $|\FPe(G,a,c)| = 2^{n}-1$ and
        $$
        \FPe(G,a,c) = \{\, \sigma \subseteq [n] \mid \sigma \ne \emptyset \,\},
        $$
        \item[\textnormal{(ii)}]  If  $1 < c < a + 1 $ (moderate inhibition), then $|\FPe(G,a,c)| = 1$ and
        $$
        \FPe(G,a,c) = \{\, \{n\} \,\},
        $$
        where $\{n\}$ is the e-support for the  singleton fixed point.
        \item[\textnormal{(iii)}]  If $0<c<1$ (weak inhibition), then $|\FPe(G,a,c)| = 1$ and 
        $$
        \FPe(G,a,c) = \{\, [n] \,\},
        $$
        where $[n]$ is the full e-support.
    \end{itemize}
\end{restatable}

\begin{figure}
    \centering
    \includegraphics[width=0.95\linewidth]{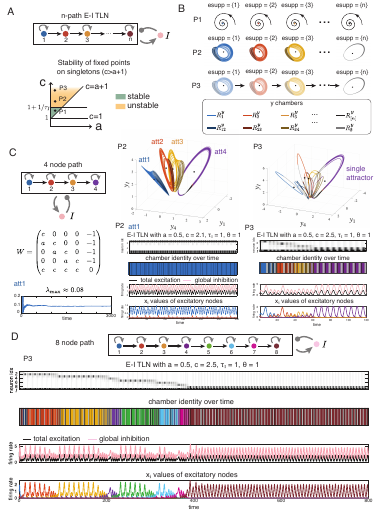}
     \caption{Emergent dynamics under strong inhibition for an E-I TLN whose underlying graph is a directed $n$-path.
    (A) Network architecture satisfying Dale's law, together with the stability conditions for singleton fixed points in the strong inhibition regime.
    (B) Three dynamical states (shown by P1, P2, and P3) in the strong inhibition regime.
    Colors in each attractor indicate different chambers, each governed by a linear subsystem. The final attractor around node $n$ is drawn in black, since the color associated with node $n$ is undetermined.
    (C) Illustration of P2 and P3 in a $4$-path E-I TLN. The underlying $4$-path graph and its connectivity matrix $W$ are shown on the left. Firing rates over time are shown in grayscale, and trajectories are shown in $y$-space. Attractor 1 is chaotic, with largest Lyapunov exponent $\lambda_{\max}\approx 0.08$.
    (D) Illustration of P3 in an $8$-path E-I TLN.}
    \label{fig:npath_high_inhibition}
\end{figure}

\begin{proof}
    \emph{Case (i): $c > a + 1$}.
    For any nonempty subset $\sigma \subseteq [n]$, the induced subgraph $G|_\sigma$ of the path $G$ is a disjoint union of one or more directed paths. Thus, the on-neuron condition is satisfied by Rule~5, and the off-neuron condition is satisfied by Rule~6 since for all $k \in [n]\setminus \sigma$, $d_k \le 1$. Hence, every nonempty $\sigma\subseteq [n]$ e-supports a fixed point, that is, $\sigma \in \FPe(G,a,c)$. 
    Since the empty set cannot be an e-support in  E–I TLNs due to the external input $\theta>0$, we have $\FPe(G,a,c) = \{\sigma \subseteq [n] \mid \sigma \ne \emptyset\}$.
    
    \emph{Case (ii): $1 < c < a + 1$}. By the definition of domination, we know node $2$ graphically dominates node $1$ in $G$, so by Rule~1, we have
    $
    \FPe(G,a,c) = \FPe(G|_{[n]\setminus\{1\}},a,c).
    $
    Next, within the induced subgraph $G|_{[n]\setminus\{1\}}$, node $3$ dominates node $2$, and again by Rule~1,
    $
    \FPe(G|_{[n]\setminus\{1\}},a,c) = \FPe(G|_{[n]\setminus\{1,2\}},a,c).
    $
    Repeating this argument recursively for nodes $3,4,\dots,n-1$, we obtain
    $$
    \FPe(G,a,c) = \FPe(G|_{\{n\}},a,c).
$$
    The remaining subgraph $G|_{\{n\}}$ is a singleton, which means $\FPe(G,a,c) \subseteq \{\{n\}\}$.  Since $G|_{\{n\}}$ has uniform in-degree $0$, Rules~3 and~4 imply that $\{n\} \in \FPe(G,a,c)$.  Therefore, when \jz{ $1 < c < a + 1$}, we have $\FPe(G,a,c) = \{\{n\}\}$.
    
    \emph{Case (iii): $c<1$}. By Rule~5, we already know that $[n] \in \FPe(G,a,c)$.  
    Now suppose, for contradiction, that there exists another e-support $\sigma \in \FPe(G,a,c)$ with $\sigma \neq [n]$.
    Assume $1 \notin \sigma$. Then by the path structure of $G$, each node $i$ weakly dominates node $i+1$ for $i = 1,2,\dots,n-1$.  
    Applying Rule~2 successively, we obtain 
    $2 \notin \sigma$, $3 \notin \sigma$, and so on, eventually leading to $\sigma = \emptyset$.  
    Since the e-support cannot be empty, this yields a contradiction.  
    Hence, any e-support $\sigma$ must contain node 1.
    If $1 \in \sigma$, then in the path structure every node $i$ ($i = 2,3,\dots,n$) weakly dominates node 1.  
    By Rule~2, this implies $2,3,\dots,n \in \sigma$, and therefore $\sigma = [n]$.
    Combining the two arguments, we conclude that when $c < 1$, $\FPe(G,a,c) = \{[n]\}$.
\end{proof}

Recall that $[n]=\{1,2,\dots,n\}$, while $\{n\}$ denotes the singleton set containing the excitatory node $n$. The cases $c=a+1$ and $c=1$ are excluded here to avoid degeneracies; the reason for excluding these parameter values can already be seen in the $2$-path E-I TLN (see the Appendix). From Theorem~\ref{thm:path_supports} it is convenient to denote the three parameter regions $c>a+1$, $1<c<a+1$, and $0<c<1$ as the strong, moderate, and weak inhibition regimes, respectively. All three regimes require $a>0$. The moderate inhibition regime corresponds to the CTLN regime and therefore has the same fixed-point structure as CTLNs.

For the $n$-path ($n\ge 2$), sequential chaotic oscillations are observed around unstable singleton fixed points in the strong inhibition regime $c>a+1$. By Lemma~\ref{lem:unique fp} in the Appendix, the stability of a singleton fixed point with support ${k,I}$ is determined by the reduced Jacobian

\begin{equation*}
J =
\begin{pmatrix}
1 & 0\\
0 & \frac{1}{\tau_I}
\end{pmatrix}
\begin{pmatrix}
c-1 & -1\\
c & -1
\end{pmatrix}
=
\begin{pmatrix}
    c - 1 & -1 \\
    \frac{c}{\tau_I} & -\frac{1}{\tau_I}
\end{pmatrix},
\end{equation*}
whose eigenvalues have already been analyzed in the singleton E-I TLN.  This yields a simple stability criterion in E-I TLNs: any fixed point with a singleton e-support is stable when $c < 1 + 1/\tau_I$ and becomes unstable when $c > 1 + 1/\tau_I$. 

Depending on the stability of the singleton fixed points, the strong inhibition regime is divided into two regions. Within these two regions, however, the system exhibits three distinct dynamical states, illustrated by the sample points P1, P2, and P3 in Figure~\ref{fig:npath_high_inhibition}. \jz{ At P1 ($c<1+1/\tau_I$), the network admits a set of coexisting steady
states, which are stable foci and distinguished by their excitatory supports}.

\jz{At P2 ($c>1+1/\tau_I$), the singleton fixed points become unstable and give rise to $n$ dynamical attractors. The first $n-1$ attractors appear to be chaotic, with largest Lyapunov exponents $\lambda_{\max}\approx 0.08$, estimated numerically using a standard two-trajectory renormalization method~\cite{benettin1980lyapunov, wernecke2019chaos}. Note that these chaotic behaviors emerge when the fixed points at P1 become unstable (at P2). However, the transition here is not a standard Hopf bifurcation because the dynamics are linear in a neighborhood of the fixed points.} The phase portraits of these attractors are shown in $y$-space rather than $x$-space in order to reveal their full structure (see Figure~\ref{fig:2path} of the $2$-path E-I TLN for a comparison of the two representations). Moreover, the attractors are colored by chambers, with each chamber governed by a linear subsystem. Although each chaotic attractor shown in Figure~\ref{fig:npath_high_inhibition}B and C traverses only two chambers, in other cases chaotic attractors may span more than two chambers while preserving the same M\"obius-like geometric structure (see Figure~\ref{fig:2path_chaos}).

P3 lies in the same parameter region ($c>1+1/\tau_I$) as P2, but moving from P2 to P3, a clear bifurcation occurs around the first $n-1$ singleton fixed points, where the corresponding chaotic attractors become attractor ruins. These attractor ruins transiently capture trajectories and decay sequentially in the path order determined by the graph, $1\to 2\to \cdots \to n$. As a result, the system ultimately converges to the attractor associated with the terminal node $n$. We refer to this behavior as sequential chaotic oscillations (SCOs), since it consists of chaotic dynamics driven by E-I oscillations, with transition order determined by the graph.

\jz{Based on our simulations, we conjecture that the transition from P2 to P3 may be analogous to a boundary crisis in the Lorenz system \cite{lorenz2017deterministic}, in which a chaotic attractor loses its attracting property and is replaced by a chaotic saddle \cite{grebogi1983crises}. A chaotic saddle is a non-attracting chaotic invariant set: trajectories can remain in its neighborhood for a long time and exhibit irregular, sensitive dynamics before eventually escaping. Under this interpretation, the faded orbits at P3 may represent chaotic saddles, while the corresponding attractor ruins may be understood as long-lived transient trajectories that linger near these saddles. The metastable dynamics at P3 may therefore be interpreted as sequential visits to a collection of chaotic saddles. We emphasize, however, that this interpretation is conjectural and requires further numerical and analytical validation.}

SCOs become clearer in longer paths. Examples for the $4$-path and $8$-path E-I TLNs are shown in Figure~\ref{fig:npath_high_inhibition}C and \ref{fig:npath_high_inhibition}D. Compared with the $4$-path case, the $8$-path E-I TLN exhibits more heterogeneous chaotic dynamics: the flow spends significantly longer times near singleton fixed points $2$ and $3$ than near the other nodes. These extended oscillatory episodes resemble bursting behavior, with waveforms that are qualitatively distinct from the shorter oscillations observed near the remaining nodes. Notably, SCOs are sensitive to initial conditions, yet the order of transitions remains fixed, as shown in Figure~\ref{fig:background}C.

For completeness, the fixed points and associated attractors in the moderate and weak inhibition regimes are shown in the Appendix (Figure~\ref{fig:npath_other_inhibition}). In both regimes, there is a unique fixed point. When this fixed point loses stability, the system transitions to singleton E--I oscillations in the moderate inhibition regime and to global E--I oscillations in the weak inhibition regime. Interestingly, in both dynamical states of the weak inhibition regime---whether at a steady state or in an E--I oscillatory state---the relative activity levels are shaped by the asymmetry of the underlying graph: nodes closer to the terminal node exhibit higher activity. 

\subsection{E-I TLNs on $n$-cycle}
As in the path case, Theorem~\ref{thm:cycle_supports} is obtained from the six graph rules. It gives the fixed point structure of E-I TLNs on cycles and the corresponding breakdown of parameter regimes, which helps us understand the different dynamical behaviors.

\begin{restatable}[Fixed points for cycles]{theorem}{thmtwo}\label{thm:cycle_supports}
    Let $G$ be the $n$-cycle with $n \geq 3$ ($1\to2\to\dots \to n \to 1$). Then, for any nondegenerate E-I TLN with $a>0$, $c>0$, the collection of fixed point e-supports $\FPe(G,a,c)$ satisfies:
    \begin{itemize}
        \item[\textnormal{(i)}]  If $c > a+1$ (strong inhibition), then $|\FPe(G,a,c)| = 2^n-1$ and
        $$
        \FPe(G,a,c) = \{\, \sigma \subseteq [n] \mid \sigma \ne \emptyset \,\},
        $$
        \item[\textnormal{(ii)}]  If $\tfrac{a - 1}{n-1} < c < a + 1 $, then $|\FPe(G,a,c)| = 1$ and the unique e-support is the full e-support, 
        $$
        \FPe(G,a,c) = \{\, [n] \,\},
        $$
        where $[n] = \{1,2,\dots,n\}$.
        \item[\textnormal{(iii)}]  If $c < \tfrac{a - 1}{n-1}$, then 
        $$
        |\FPe(G,a,c)| = 0.
        $$
    \end{itemize}
\end{restatable}

\begin{proof}
   A first observation is that, since $G$ has  uniform in-degree~$1$, Rule~3 implies that the full-support fixed point exists if and only if $c > (a-1)/(n-1)$. Consequently, $[n] \in \FPe(G,a,c)$ in cases~(i) and~(ii), but not in case~(iii).

   \emph{Case (i): $c > a + 1$.}
    We now consider the nonempty proper subset $\sigma \subsetneq [n]$.  
    Since $G$ is an $n$-cycle, the induced subgraph $G|_\sigma$ is a disjoint union of one or more directed paths.  Similar to the path case, by Rules~5 and~6,  $\sigma \in \FPe(G,a,c)$.
    Since the empty set cannot be a support, we conclude that
    $$
    \FPe(G,a,c) = \{\, \sigma \subseteq [n] \;\mid\; \sigma \ne \emptyset \,\}.
    $$
    
    \emph{Cases (ii) and (iii): $c<a+1$ and $c \neq 1$}. Next, we prove the last two cases by showing that no proper subset $\sigma \subsetneq [n]$ can belong to $\FPe(G,a,c)$ if $c<a+1$. 

    When $1< c<a+1$, using graphical domination (Rule~1), we see that for any e-support $\sigma \subsetneq [n]$ in $\FPe(G,a,c)$, the induced subgraph $G|_{\sigma}$ must be an independent set. Indeed, if $\sigma$ were not an independent set, then there must exist two nodes $j,j+1 \in \sigma$ such that $j-1 \notin \sigma$ and $j \to j+1$, where $n+1 \equiv 1$.
    In the induced subgraph $G|_\sigma$, $j$ is dominated by $j+1$ and graphical domination implies
    $$\FPe(G|_{\sigma},a,c) = \FPe(G|_{\sigma\setminus j},a,c).$$
    But since $j \in \sigma$, it follows that $\sigma \notin \FPe(G|_{\sigma},a,c)$, a contradiction. Hence every admissible $\sigma \subsetneq [n]$ must be an independent set. But as shown by Rule~4, independent sets in cycles require $c\ge a+1$ to become an e-support. Thus, no proper subset can belong to $\FPe(G,a,c)$ for any $1 < c < a + 1$.

    When $c < 1$, suppose $\sigma \neq \emptyset$ is a fixed point e-support, that is, $\sigma \in \FPe(G,a,c)$. Pick any node $j \in \sigma$. For an $n$-cycle, every node $i \in [n]$ weakly dominates its successor $i+1$ (with $n+1 \equiv 1$). Thus, $j-1$ weakly dominates $j$, and Rule~2 implies $j-1 \in \sigma$.  
    Repeating this argument inductively gives $j-2 \in \sigma$, $j-3 \in \sigma$, and so on, until all nodes are included.  
    Hence $\sigma = [n]$. Therefore, no proper subset can be a fixed point e-support when $c < 1$.
    
    Combining these two arguments with Rule~3, since every cycle has uniform in-degree $1$, we conclude that if $(a - 1)/(n - 1) < c < a + 1 $, 
    $\FPe(G,a,c) = \{[n]\}$ and if $c < (a - 1)/(n - 1)$, $\FPe(G,a,c) = \emptyset$.
\end{proof}

\begin{figure}
    \centering
    \includegraphics[width=1\linewidth]{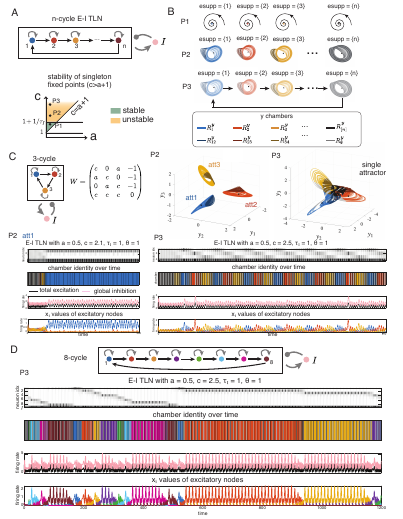}
    \caption{Emergent dynamics around singleton fixed points for an E-I TLN whose underlying graph is a directed $n$-cycle.
    (A) Network architecture and stability conditions for singleton fixed points.
    (B) Three parameter choices associated with singleton fixed points.
    \textbf{P1}: All $n$ singleton fixed points are stable.
    \textbf{P2}: Each singleton fixed point loses stability and generates a chaotic attractor, yielding $n$ chaotic attractors in total.
    \textbf{P3}: Sequential chaotic oscillations: all chaotic attractors lose stability and decay sequentially along the cycle $(1 \to 2 \to \cdots \to n \to 1)$.
    Colors indicate chambers of the E-I TLN state space, with each chamber corresponding to a distinct linear subsystem.
    (C) Illustration of P2 and P3 for a $3$-cycle E-I TLN.
    (D) Illustration of P3 for an $8$-cycle E-I TLN.}
    \label{fig:ncycle_high_inhibition}
\end{figure}

The fixed-point structure of the $n$-cycle E-I TLN differs from that of the path case in the moderate and weak inhibition regimes. When $0<c<(a-1)/(n-1)$, no fixed point exists and the activity blows up in our simulations. In the remaining parameter regions, the observed attractors arise only near singleton and full-support fixed points. Around singleton fixed points, we find chaotic attractors and SCOs, while around the full-support fixed point, we observe E-I oscillations, CTLN-like oscillations, and flower-like quasi-periodic attractors.

\paragraph{Singleton fixed points.}
We begin with singleton fixed points. As in the path case, the system exhibits three distinct types of dynamics in the two parameter regions $c < 1 + 1/\tau_I$ and $c > 1 + 1/\tau_I$. Representative examples are shown at P1, P2, and P3 in Figure~\ref{fig:ncycle_high_inhibition}. The dynamics at P1 and P2 are analogous to those observed in the $n$-path case. The main difference is that the chaotic attractor around the singleton $\{n\}$ more closely resembles a Möbius ring, due to the additional connection from node $n$ to node $1$.

At P3, in contrast to the unique chaotic attractor around the singleton fixed point $\{n\}$ in the path case, the E-I TLN on cycles exhibits high-dimensional chaos associated with all $n$ singleton fixed points. This occurs because the chaotic attractor associated with node~$n$ in the $n$-cycle case also becomes an attractor ruin due to the bifurcation. As a result, the attractor ruins decay cyclically through all nodes,
$1 \;\to\; 2 \;\to\; \cdots \;\to\; n \;\to\; 1$,
as illustrated in Figure~\ref{fig:ncycle_high_inhibition}. This switching process never terminates, indicating that the sequence of transitions among attractor ruins itself forms a high-dimensional global chaotic attractor. This behavior is reminiscent of the Lorenz attractor~\cite{lorenz2017deterministic}: whereas the Lorenz attractor is a butterfly with two wings, the chaotic attractor in the $3$-cycle E-I TLN is a butterfly with three wings, as shown in Figure~\ref{fig:ncycle_high_inhibition}C. In this sense, the Lorenz attractor may be viewed heuristically as switching between two metastable regions.
In Figure~\ref{fig:ncycle_high_inhibition}D, we show SCOs for the $8$-cycle E-I TLN, where the dwell times at individual nodes vary significantly and irregularly. Moreover, SCOs do not require fine-tuning: Figure~\ref{fig:3cycle_sco} shows additional $(a,c)$ parameter choices for which SCOs occur in the $3$-cycle E-I TLN.

\jz{Following the conjectural picture above for the path, the SCOs in the $n$-cycle may be interpreted as a cyclic sequence of chaotic saddles. Unlike the path case, the cycle has no terminal attractor, so trajectories continue switching as $1 \to 2 \to \cdots \to n \to 1$, with irregular dwell times near each chaotic saddle. This unending switching process forms a global chaotic attractor. The route to this chaos is reminiscent of that proposed for fluid-dynamic turbulence~\cite{eckhardt2007turbulence}, where chaotic dynamics are organized by coexisting unstable invariant states. In our setting, these unstable invariant states are chaotic saddles.}

\paragraph{The full-support fixed point (z-mode and mean mode analysis).}

Here we focus on the full-support fixed point, whose excitatory support is the full set $[n]$ and which exists for
$c > (a-1)/(n-1)$. By analyzing its Jacobian matrix, we find that this fixed point is stable if
$$
c < 1 - a \cos\!\left(\tfrac{2\pi}{n}\right)
\quad \text{and} \quad
a + c < 1 + \tfrac{1}{\tau_I}.
$$
Outside this stability region ($c > 1 - a \cos\tfrac{2\pi}{n}$ or $a + c > 1 + \tfrac{1}{\tau_I}$), 
different types of attractors are observed across the resulting parameter regimes. To distinguish the corresponding parameter regions, it is useful to separate two complementary aspects of the network activity. The first captures relative differences among the excitatory nodes and is essential for identifying sequential or asymmetric activity patterns. The second describes the collective activity of the excitatory population together with the inhibitory neuron, and governs globally synchronized oscillations. These two components are formalized as the \emph{$z$-mode} and the \emph{mean mode}.

\begin{definition}[$z$-mode and mean mode]
Consider an E-I TLN with $n$ excitatory nodes and one inhibitory node. We define the difference variables
$$
z_j = x_{j+1} - x_j, \quad j = 1, \dots, n-1,
$$
and the mean excitatory variable 
$$
x_E = \sum_{j=1}^{n} x_j.
$$
The $z$-mode is the subsystem governing the dynamics of $z = (z_1, \dots, z_{n-1})$, and the mean mode is the subsystem governing the dynamics of the mean variables $(x_E, x_I)$.
\end{definition}
\begin{remark}
The variable $z_n=x_1-x_n$ is not included, since
$z_n=-\sum_{j=1}^{n-1} z_j$. 
For example, for a $3$-cycle E--I TLN
(Figure~\ref{fig:cycle_full_support}B), the $z$-mode is described by
$z_1=x_2-x_1$ and $z_2=x_3-x_2$. 
The remaining difference variable $z_3=x_1-x_3$ is omitted because it is
determined by $z_3=-z_1-z_2$.

The terminology for the mean mode is motivated by a mean-field perspective.
However, $x_E$ is not normalized by $n$, and hence represents the total
excitatory activity rather than the average excitatory activity. This choice
allows for a more direct comparison with the total inhibitory activity.
\end{remark}

Although the definition of the $z$-mode depends on the ordering of the nodes in $[n]$, the choice of ordering does not affect the local behavior or the corresponding stability conclusions. Moreover, the $z$-mode and the mean mode decouple in the full-support chamber when the graph $G$ has uniform in-degree and out-degree. Details and proofs concerning the decoupling are given in Section~\ref{sec:z-mode}. Since any $n$-cycle has uniform in-degree $1$ and out-degree $1$, the $z$-mode and mean mode decouple in the chamber $R_{[n]}$, so that the full-support fixed point decomposes into a fixed point $z^*=0$ of the $z$-mode and a fixed point $(x_E^*,x_I^*)$ of the mean mode. Moreover, they also decouple in the chamber $R_\emptyset$. Even though the full system has no fixed point in the chamber $R_{\emptyset}$, the $z$-mode subsystem still has a fixed point there. These properties are summarized in the following theorem.
\begin{restatable}{theorem}{thmthree}\label{thm:cycle_2_modes}
Let $G$ be an $n$-cycle ($n\geq 3$) and suppose it is the underlying graph of a nondegenerate E-I TLN. Then the corresponding $z$-mode and mean mode decouple in the chambers $R_{[n]}$ and $R_\emptyset$. 
Moreover, if there exists a full-support fixed point $x^*$, then
\begin{itemize}
    \item In the chamber $R_{[n]}$, the $z$-mode has a fixed point $z^* = 0$, which is stable when $c < 1 - a\cos\!\left(\tfrac{2\pi}{n}\right)$ and unstable when $c > 1 - a\cos\!\left(\tfrac{2\pi}{n}\right)$. 
    \item In the chamber $R_{[n]}$, the mean mode has a fixed point $(x^*_E,x^*_I) = (\sum_{i=1}^{n} x_i^*,x^*_I)$, which is stable when $a+c < 1 + \tfrac{1}{\tau_I}$ and unstable when $a+c > 1 + \tfrac{1}{\tau_I}$.
    \item In the chamber $R_\emptyset$, the $z$-mode has a fixed point $z^* = 0$, which is always stable.
\end{itemize}
\end{restatable}

From Theorem~\ref{thm:cycle_2_modes}, the original system decomposes into $z$-mode and mean mode in the chamber $R_{[n]}$. 
Here, decoupling means that the equations for $z=(z_1,\dots,z_{n-1})$ involve only $z$, while the equations for $(x_E,x_I)$ involve only $(x_E,x_I)$, so the two subsystems can be analyzed independently within the chamber.
As a consequence, the full-support fixed point of the E-I TLN is stable if and only if both subsystems admit stable fixed points, namely when $c < 1 - a\cos(2\pi/n)$ and $a+c < 1 + 1/\tau_I$. It therefore remains to analyze the dynamics arising when one of these two modes loses stability.

\begin{figure}
    \centering
    \includegraphics[width=1\linewidth]{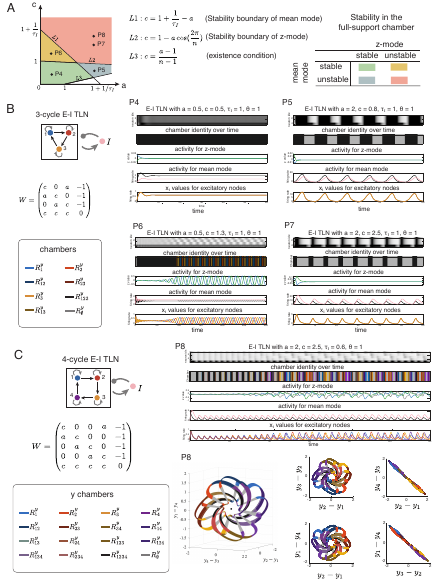}
   \caption{
    Emergent dynamics around the full-support fixed point for an E-I TLN whose underlying graph is an $n$-cycle.
    (A) Four regions determined by the stability conditions of the $z$-mode and mean mode fixed points in the chamber $R_{[n]}$.
    (B) Illustration of the four parameter choices (P4-P7) for the E-I TLN on a $3$-node cycle, corresponding to the four regions shown in panel~A. 
    (C) Flower-like quasi-periodic \jz{behavior} in a 4-cycle E–I TLN when both the $z$-mode and the mean mode are unstable.
    }
    
    \label{fig:cycle_full_support}
\end{figure}

First, consider the case where the mean mode becomes unstable while the $z$-mode remains stable.
By Theorem~\ref{thm:cycle_2_modes}, the $z$-mode converges to the fixed point $z^* = 0$ in both chambers $R_{[n]}$ and $R_\emptyset$.
As shown in P5 of 3-cycle E-I TLN in Figure~\ref{fig:cycle_full_support}B, the instability of the mean mode gives rise to E–I oscillations that traverse between the chambers $R_{[n]}$ and $R_\emptyset$.
Since the $z$-mode is stable in both chambers, all excitatory nodes share identical activity levels throughout the oscillation, resulting in synchronized excitatory dynamics.

Next, consider the complementary case, in which the mean mode remains stable while the $z$-mode becomes unstable.
In this regime, the difference variables $z_j$ associated with excitatory activity exhibit oscillatory behavior, while the mean activity stays near its equilibrium value.
This leads to CTLN-like oscillations characterized by sequential activation patterns among the excitatory nodes.
The small fluctuations in the mean activity observed in P6 of Figure~\ref{fig:cycle_full_support}B arise from the system switching between different chambers, rather than from an intrinsic instability of the mean mode itself.

Finally, when both the mean mode and the $z$-mode become unstable, the resulting dynamics are more complex.
In one scenario (P7 in Figure~\ref{fig:cycle_full_support}B), the instability of the mean mode generates E–I oscillations that traverse the chambers $R_{[n]}$ and $R_\emptyset$.
Although the fixed point $z^* = 0$ of the $z$-mode is unstable in the chamber $R_{[n]}$—so that the $z$-mode initially diverges from zero—it remains stable in the chamber $R_\emptyset$.
As a result, if the trajectory enters $R_\emptyset$, the $z$-mode can converge back to $z^* = 0$.
Once the trajectory reaches $z^* = 0$, it remains there, since $z^* = 0$ is still a fixed point of the $z$-mode in the chamber $R_{[n]}$, even though it is unstable in this chamber. Consequently, the system exhibits synchronized E–I oscillations, similar to those observed in P5.

In contrast, a different behavior (P8), which is also an E-I oscillation, is observed for the $4$-cycle E-I TLN even in the same parameter region as P7, where both modes are unstable. In this case, the $z$-mode does not converge to $z^*=0$ sufficiently rapidly when the trajectory passes through the chamber $R_\emptyset$. As a result, the $z$-mode alternates between phases of convergence and divergence under the E-I oscillation, giving rise to quasi-periodic dynamics; see P8 in Figure~\ref{fig:cycle_full_support}. The resulting attractor also has an interesting geometry. For the coordinate differences $y_2-y_1$, $y_3-y_2$, $y_4-y_3$, and $y_1-y_4$, the relation between $y_2-y_1$ and $y_4-y_3$ is nearly linear, suggesting that $y_4+y_2 \approx y_3+y_1$ is approximately conserved along the trajectory. A similar relation holds between $y_3-y_2$ and $y_1-y_4$. Other combinations of coordinate differences reveal a distinctive flower-like structure underlying the E-I oscillations, and we therefore refer to it as a flower-like attractor.

\subsection{Effect of $\tau_I$ on E-I TLN dynamics}

Here we investigate how $\tau_I$ shapes the dynamics of E-I TLNs, since all of our theoretical results depend on the inhibitory timescale $\tau_I$. To this end, we consider the limiting cases $\tau_I \to 0$ and $\tau_I \to \infty$ in Figure~\ref{fig:effect_tau_I}. For singleton fixed points in the strong and moderate inhibition regimes, the limit $\tau_I \to 0$ stabilizes all singleton fixed points, thereby eliminating the surrounding dynamical attractors. In contrast, the limit $\tau_I \to \infty$ destabilizes all singleton fixed points. For the full-support fixed point in the $n$-cycle case, the limit $\tau_I \to 0$ ensures that the mean mode remains stable. As a result, once the $z$-mode becomes unstable and the trajectory does not blow up, the system exhibits CTLN-like oscillations, as also observed in~\cite{curto2025graphical}. This mirrors the behavior of the matched CTLN shown in Figure~\ref{fig:effect_tau_I}C. By contrast, when $\tau_I \to \infty$, the CTLN-like oscillations disappear, and the system instead exhibits E-I oscillations.

Finally, we summarize the long-term behaviors of E-I TLNs on $n$-paths and $n$-cycles across inhibition regimes and the inhibitory timescale $\tau_I$ in Tables~\ref{tab:paths} and~\ref{tab:cycles}. Here we use singleton E-I oscillations to refer to E-I oscillations between a single excitatory node and the inhibitory node. \jz{Synchronized} E-I oscillations refer to E-I oscillations where the excitatory nodes have the same firing phase.
Thus, both singleton E-I oscillations and sequential chaotic oscillations (SCOs) are unsynchronized E-I oscillations. For the E-I TLNs on paths, the above E-I oscillations only appear for slow $\tau_I$ while fast $\tau_I$ always \jz{stabilizes} the fixed point.

\begin{table}[htbp]
\footnotesize
\renewcommand{\arraystretch}{1.3}
\caption{E-I TLNs on paths}\label{tab:paths}
\begin{center}
  \begin{tabular}{c|c|c|c} \hline
    $G$ is an $n$-path &  strong inhibition & moderate inhibition  & weak inhibition \\ \hline
   $\FPe(G,a,c)$ & $\{ \sigma \subseteq [n] \mid \sigma \ne \emptyset \}$ & $\{\, \{n\} \,\}$ 
         & $\{\, [n] \,\}$ \\ \hline
    slow $\tau_I$ & singleton E-I oscillations or SCOs & singleton E-I oscillations & synchronized E-I oscillations\\
    fast $\tau_I$  & stable fixed points on $\{i, I\}$, $i \in [n]$ & stable fixed point on $\{n, I\}$ & stable fixed point on $[n] \cup \{I\}$ \\ \hline
  \end{tabular}
\end{center}
\end{table}

For the full-support fixed point of an E-I TLN on an $n$-cycle, the stability of the mean mode depends strongly on $\tau_I$, whereas that of the $z$-mode does not. When $\tau_I$ is sufficiently fast, the mean mode is always stable, and the dynamics are determined by the stability of the $z$-mode: the system exhibits either a steady state or CTLN-like oscillations. By contrast, when $\tau_I$ is sufficiently slow, the mean mode is unstable for most parameter values, although it remains stable in a small region of parameter space. However, the region for CTLN-like oscillations, where the mean mode is stable and the $z$-mode is unstable, disappears; see Figure~\ref{fig:effect_tau_I}C.  As a result, CTLN-like oscillations no longer occur. Instead, if the mean mode is unstable, the system exhibits synchronized E-I oscillations or a flower-like attractor; if it is stable, the dynamics converge to the fixed point on $[n]\cup\{I\}$.

\begin{table}[htbp]
\footnotesize
\renewcommand{\arraystretch}{1.3}
\caption{E-I TLNs on cycles}\label{tab:cycles}
\begin{center}
  \begin{tabular}{c|c|c} \hline
    $G$ is an $n$-cycle &  strong inhibition & moderate and weak inhibition \\ \hline
   $\FPe(G,a,c)$ & $\{ \sigma \subseteq [n] \mid \sigma \ne \emptyset \}$ & $\{\, [n] \,\}$ for $c>(a-1)/(n-1)$
          \\ \hline
    slow $\tau_I$ & singleton E-I oscillations or SCOs & E-I oscillations or stable fixed point  on $[n] \cup \{I\}$ \\
    fast $\tau_I$  & stable fixed points on $\{i, I\}$, $i \in [n]$ & CTLN-like oscillation or stable fixed point  on $[n] \cup \{I\}$ \\ \hline
  \end{tabular}
\end{center}
\end{table}

\begin{figure}
    \centering
    \includegraphics[width=1\linewidth]{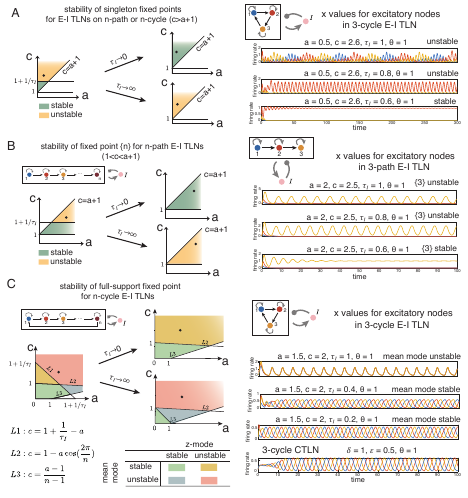}
    \caption{
    Effects of the inhibitory timescale $\tau_I$ ($\tau_I \to 0$ or $\tau_I \to \infty$) on the long-term behavior of E-I TLNs on paths and cycles.
    (A) Influence of $\tau_I$  on the stability of singleton fixed points in the strong inhibition regime for E-I TLNs on $n$-paths and $n$-cycles. A 3-cycle E-I TLN is shown as an example.
    (B) Same as (A), but in the moderate inhibition regime and only for E-I TLNs on $n$-paths. Examples are shown for a 3-path E-I TLN.
    (C) Influence of $\tau_I$ on the stability of the $z$-mode and mean mode fixed points, obtained from the mode decomposition of the full-support fixed point, for E-I TLNs on $n$-cycles. A 3-cycle E-I TLN and a CTLN with matched parameters are shown as examples on the right.
    In each panel, the parameters $a$ and $c$ are fixed while $\tau_I$ is varied.
    }
    \label{fig:effect_tau_I}
\end{figure}

\section{Discussion}
\label{sec:discussion}
In this paper, we analyze the fixed point structure of nondegenerate E-I TLNs on paths and cycles over the entire positive $(a,c)$-parameter plane. 
In particular, we develop five new graph rules for the strong- and weak-inhibition regimes. 
Among these, Rule 2 provides a support-exclusion criterion in the weak-inhibition regime, while Rules 3 and 4 yield on- and off-neuron conditions for subgraphs with uniform in-degree.
Rule 5 gives an on-neuron condition for disjoint unions of directed paths, and Rule 6 provides a sufficient off-neuron condition for arbitrary subgraphs.
Combining these results with domination theory (Rule 1) in the moderate-inhibition regime, we characterize the complete fixed point structure of nondegenerate E--I TLNs on $n$-paths and $n$-cycles; see Theorems~\ref{thm:path_supports} and~\ref{thm:cycle_supports}.

Based on the fixed points obtained above, we observe some interesting attractors in numerical simulations (MATLAB \texttt{ode45}). Notably, associated with the full-support fixed point of E-I TLNs on cycles, we find several types of attractors, including synchronized E-I oscillations, CTLN-like oscillations, and flower-like attractors. To distinguish the parameter regions in which these behaviors occur, we introduce the $z$-mode, together with the mean mode, as a new analytical tool. While the mean mode captures total excitatory and inhibitory activity, the $z$-mode captures differences in activities among excitatory nodes. The stability of these two modes allows us to classify the observed attractors.

In addition, around the singleton fixed points in the strong inhibition regime, we observed several chaotic behaviors, including singleton E-I oscillations and sequential chaotic oscillations (SCOs). The geometry of these attractors resembles that of classical chaotic systems.
 For example, the phase portrait of a singleton E-I oscillation in $y$-space exhibits a M\"obius-like geometry reminiscent of the R\"ossler attractor~\cite{rossler1976equation}. Likewise, SCOs in $n$-cycle E--I TLNs display a multi-wing structure
reminiscent of the Lorenz attractor~\cite{lorenz2017deterministic,tucker2002rigorous}.
While the Lorenz attractor is a two-winged ``butterfly'', the $3$-cycle and $8$-cycle E-I TLNs exhibit three- and eight-winged attractors, respectively.
In these networks, trajectories switch among the wings in a fixed order determined by the underlying graph, but with unpredictable dwell times.

More broadly, the dynamical behaviors discussed here, including singleton E-I oscillations, SCOs, and flower-like attractors, can all be interpreted as unsynchronized E-I oscillations. This notion differs from the strong, sustained, and regular E-I oscillations~\cite{wang1996gamma,wang2010neurophysiological}, and from the desynchronized E-I oscillations in spiking networks~\cite{borgers2003synchronization,brunel2008sparsely}. In our setting, E-I oscillations need not be regular or persistent; they may be transient, irregular, and even chaotic, as in the case of SCOs.
This perspective is consistent with recent experimental work by Battaglia and colleagues, who reported weak, weird, and transient oscillations in mouse hippocampal LFPs~\cite{douchamps2024gamma} and human EEG~\cite{bahuguna2025interdependence}. These studies suggest that the variability and ``messiness'' of such oscillations are not simply noise-induced, but may instead emerge from self-organized dynamics and mediate structured information processing~\cite{palmigiano2017flexible}. In this sense, the irregular oscillations observed experimentally may be viewed as functionally analogous to the SCOs found in our model: chaotic dynamics generate ``messiness'', while attractor ruins provide metastable structure, thereby supporting graph-determined sequential information flow.

As a form of chaotic itinerancy, SCOs provide a natural mechanism of itinerant dynamics. 
Such dynamics are often studied in multistable systems driven by noise~\cite{miller2016itinerancy} or by specific inputs~\cite{alvarez2026attractor}. 
By contrast, SCOs arise under constant input and follow a fixed, graph-determined transition order, while still exhibiting irregular dwell times for  metastable states. Moreover, although chaotic itinerancy has often been viewed as a high-dimensional phenomenon driven by high-dimensional chaos, our results show that this behavior does not require high dimensionality: it can already appear in dimension three, as in the E-I TLN on a 2-path. In particular, this itinerancy does not need to be driven by high-dimensional chaos (see E-I TLNs on $n$-path).
Instead, SCOs in E-I TLNs on cycles suggest that the itinerant modules (attractor ruins) may serve as building blocks for complex chaos.

Finally, because the nonlinear component of TLN models is simply a threshold nonlinearity, implemented by the ReLU function, these models provide a minimal framework for studying how connectivity alone can shape neural dynamics. For
example, TLNs have previously been used to analyze how network structure constrains the emergence of oscillations in basal-ganglia circuits~\cite{zang2024structural}. In the same spirit, SCOs in E-I TLNs offer a possible dynamical mechanism for interpreting structured manifolds observed in neural population activity~\cite{perich2025neural}. For instance,
rotational neural dynamics in stop-signal experiments~\cite{komi2025neural},
which are interpreted as limit-cycle-like dynamics, could alternatively be viewed as motion through a sequence of metastable states or attractor ruins.
Such a mechanism naturally allows trajectories to pause near metastable states and to incorporate additional components into the sequence.

\section*{Code availability}
The code used for the numerical simulations and to reproduce the figures in this paper is publicly available at
\url{https://github.com/ccurto/seq-chaos-TLN}.

\section*{Acknowledgments}
The authors thank Professor Arvind Kumar and members of the Curto Lab for helpful discussions and valuable feedback on this work.
Special thanks to Nicole Sanderson, Yuchen (Jency) Jiang, and Joaqu\'in Casta\~neda Castro for their detailed comments and careful checking of the proofs. The authors also thank the reviewers for their constructive comments.
\bigskip

\pagebreak

\bibliographystyle{unsrt}
\bibliography{references} 

\pagebreak
\section{Appendix}
\label{sec:Appendix}

In this Appendix, we review the mathematical tools adapted from CTLNs, including chambers and fixed-point supports, and explain how they differ in the E-I TLN setting. We then develop and prove the graph rules and supporting lemmas underlying the main-text proofs of Theorems~\ref{thm:path_supports} and~\ref{thm:cycle_supports}, followed by the proof of Theorem~\ref{thm:cycle_2_modes}. Finally, we give a detailed example of the \(2\)-path E-I TLN and present supplementary figures.

\subsection{Link to CTLNs} \label{sec:link_CTLN}

As an extension of combinatorial threshold-linear networks (CTLNs), E-I TLNs share some structural properties with CTLNs, including their fixed-point structure. 
In this section, we review CTLNs, explain their relationship to E-I TLNs, and introduce several notions for E-I TLNs that are inherited from the CTLN framework. We also introduce a notion of nondegeneracy for E-I TLNs and use it to show that each support admits at most one fixed point in a nondegenerate E-I TLN.

 \begin{figure}[h]
    \centering
    \includegraphics[width=0.9\linewidth]{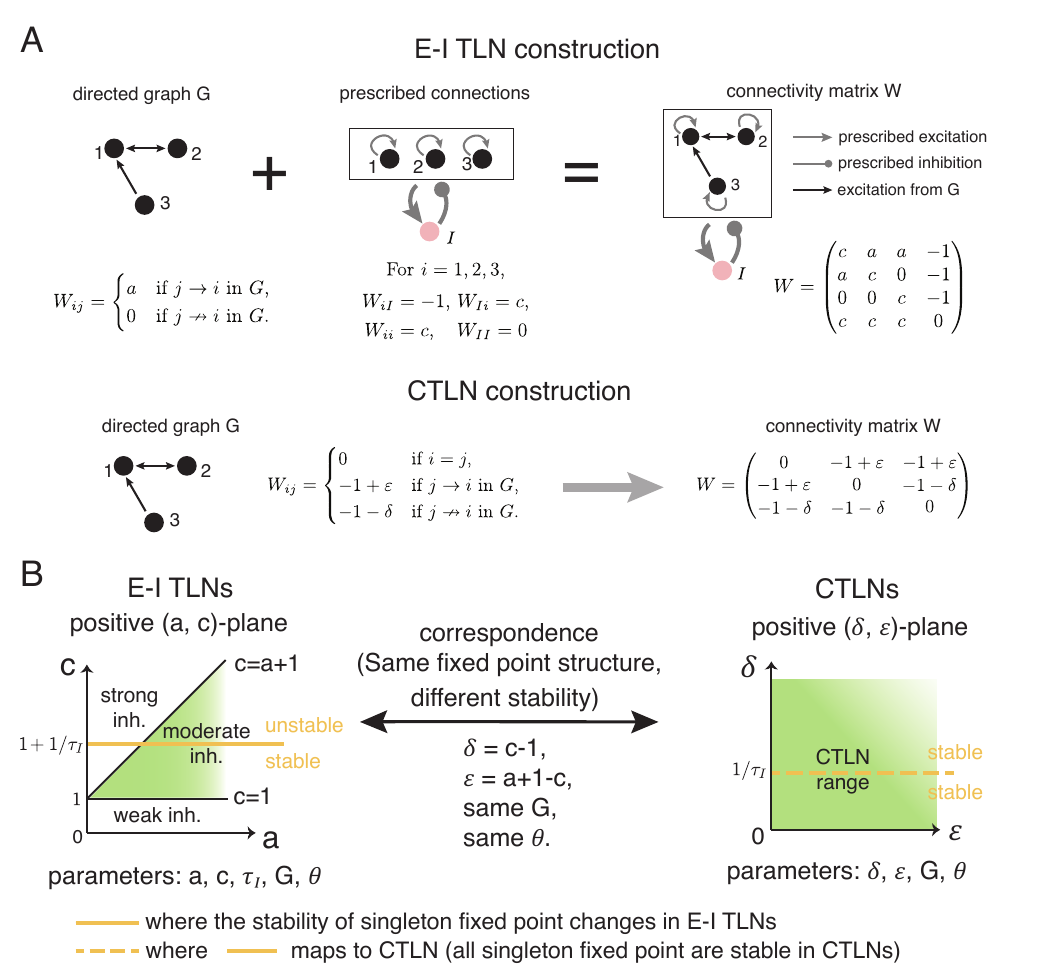}
    \caption{
    Correspondence between E-I TLNs and CTLNs.
    (A) Construction of the connectivity matrices $W$ for E-I TLNs and CTLNs from the same directed graph $G$.
    (B) Matching E-I TLNs to CTLNs in the moderate inhibition regime. Any E-I TLN in the moderate inhibition regime ($1<c<a+1$) can be matched to a CTLN with the same fixed-point structure. Moreover, under sufficiently fast inhibition ($\tau_I \ll \tau_E$), the E-I TLN exhibits long-term dynamics that closely resemble those of the matched CTLN.
    }
    \label{fig:translation from CTLN}
\end{figure}

\paragraph{CTLN setting.}
CTLNs were first introduced in \cite{curto2019fixed} and they are generated from a directed graph $G$, where each vertex corresponds to an excitatory node $i$. Their dynamics are governed by the threshold-linear system
\begin{align*}
 \frac{dx_i}{dt} &= -x_i + \left[ \sum_{j=1}^{n} W_{ij}x_j + b_i \right]_+, \qquad i = 1, \dots, n.
\end{align*}
As shown in Figure~\ref{fig:translation from CTLN}A, the connectivity matrix $W$ in CTLNs is determined by the graph $G$, where $j \to i$ indicates that there is an edge from $j$ to $i$, and $j \nrightarrow i$ indicates that there is no such edge:
\[
W_{ij} =
\begin{cases}
0 & \text{if } i=j,\\
-1+\varepsilon & \text{if } j\to i \text{ in } G,\\
-1-\delta & \text{if } j\nrightarrow i \text{ in } G.
\end{cases}
\]
All entries of $W$ are nonpositive, reflecting the assumption that excitatory nodes operate under a global “blanket of inhibition”.
Moreover, the external inputs in CTLNs are set to $b_i = \theta > 0$, for $i \in [n]$. Much of the foundational work on CTLNs is summarized in~\cite{curto2023graph,curto2019fixed}. Next, we introduce several mathematical tools originally developed for CTLNs, which can also be applied to E-I TLNs.

\paragraph{Chambers and fixed points for E-I TLNs.}
Similar to CTLNs~\cite{morrison2024diversity}, the dynamics of  E–I TLNs are bounded below, i.e., $x_i \ge 0$.  In fact, from the definition of E-I TLN equations~\eqref{eq:EI_TLN},
we have
\begin{align*}
    \frac{dx_i}{dt} \ge -x_i, \quad i = 1,2, \dots, n, I.
\end{align*}
It follows that $dx_i / dt \ge 0$  at $x_i = 0$ and $dx_i / dt > 0$ for $x_i < 0$.
Therefore, any trajectory $x(t)$ either approaches or enters the region $x_i \ge 0$ as $t \to \infty$.
Then their asymptotic dynamics are confined to the nonnegative orthant $x_i \geq 0$ for all excitatory neurons ($i = 1,2, \dots, n$) and for the inhibitory neuron ($i=I$). 

Hence, the $x$-space of interest is the nonnegative region $\mathbb{R}^{n+1}_{\ge 0}$, where the first $n$ coordinates correspond to the $n$ excitatory neurons and the last coordinate corresponds to the inhibitory neuron.  On this domain, the system is piecewise linear, with different linear subsystems governing the dynamics in different regions. To analyze how the system switches between these linear subsystems, it is convenient to introduce the variables $y$, which are the terms inside the ReLU activation. For $W_{II}=0$, $b_I=0$, and $b_i = \theta$ in E-I TLNs, they are given by
\begin{subequations}
\begin{align}
\tau_E \frac{dx_i}{dt} &= -x_i + [y_i]_+, \qquad y_i = \sum_{j=1}^{n} W_{ij}x_j + W_{iI}x_I + \theta \quad i = 1,\dots,n,\\
\tau_I \frac{dx_I}{dt} &= -x_I + [y_I]_+, \qquad y_I = \sum_{j=1}^{n} W_{Ij}x_j. 
\end{align}
\end{subequations}
A key observation is that, since  $x_j \ge 0$ and $W_{Ij} \ge 0$ in E-I TLNs, we have $y_I(x) =  \sum_{j=1}^{n} W_{Ij}x_j\ge 0$ for all $x \in \mathbb{R}^{n+1}_{\ge 0}$.  Hence, $[y_I]_+ = y_I$ for all such $x$.

For each excitatory neuron in E-I TLNs, the equation $y_i(x)=0$ defines a hyperplane in $x$-space $\mathbb{R}^{n+1}_{\ge 0}$. Since $y_I(x) \ge 0$, these hyperplanes divide $\mathbb{R}^{n+1}_{\ge 0}$ into at most $2^{n}$ regions of the form:
\begin{equation}
R_{\sigma} = \{ x \in \mathbb{R}^{n+1} \mid y_i(x) > 0 \;\; \forall i \in \sigma, \;\; y_k(x) \leq 0 \;\; \forall k \notin \sigma, \; y_I(x) \ge 0 \},
\end{equation}
where $\sigma\subseteq [n]$ and $[n] \od \{1,2,\dots,n\}$. Then, we refer to $R_{\sigma}$ as a \emph{chamber} of $\sigma$ in $x$-space.  Similarly, we define corresponding chambers in $y$-space, where $y$ is a affine transformation of $x$ and lies in $\mathbb{R}^{n+1}$. The chamber associated to $\sigma$ in $y$-space is given by:
\begin{equation}
R^y_{\sigma} = \{ y \in \mathbb{R}^{n+1} \mid y_i > 0 \;\; \forall i \in \sigma, \;\; y_k \leq 0 \;\; \forall k \notin \sigma, \; y_I \ge 0 \}.
\end{equation}
By definition, $R_{\sigma}$ and $R^y_{\sigma}$ correspond to the same space under the affine transformation between $x$ and $y$. 
Moreover, for any chamber $R_{\sigma}$ or $R^y_{\sigma}$, the sign pattern of $y_i$ is determined—so each threshold function $[y_i]_+$ becomes either $y_i$ or $0$.  As a result, the dynamics within the chamber $R_\sigma$ reduce to a linear system, which we denote by $L_\sigma$:
\begin{equation}
    L_\sigma = 
    \left\{ \left. \frac{dx_i}{dt} = - x_i + y_i \;\right|\; i \in \sigma \right\}
    \bigcup
    \left\{ \left. \frac{dx_k}{dt} = - x_k \;\right|\; k \notin \sigma, k \neq I \right\}
    \bigcup
    \left\{ \frac{dx_I}{dt} = - x_I + y_I  \right\}.
\end{equation}

 Next, we say a fixed point of an E-I TLN in $x$ space is a point $x^* \in \mathbb{R}^{n+1}$ satisfying $\left.\tfrac{dx_i}{dt}\right|_{x=x^*} = 0$ for each $i = 1,\dots,n,I$. 
In particular, this requires
\begin{equation*}
    x_i^* = [y_i(x^*)]_+ = [y_i^*]_+ \quad \text{for all } i = 1,\dots,n,I,
\end{equation*}
where $y_i^*$ denotes the value of $y_i$ evaluated at the fixed point $x^*$. Hence, $y^*$ is the corresponding fixed point in $y$-space. Now for each fixed point $x^*$, we define its \emph{support} as the set of active nodes and \emph{e-support} as the set of active excitatory nodes:
\begin{equation*}
\supp (x^*)  \od  \{\, i \in [n]\cup\{I\} \mid x_i^* > 0 \,\}, \qquad 
\esupp (x^*)  \od  \{\, i \in [n] \mid x_i^* > 0 \,\},
\end{equation*}
where $[n] = \{1,2,\dots,n\}$. 
It is then straightforward to verify that for any fixed point $x^*$ of E-I TLNs,  $\supp(x^*) = \esupp(x^*) \cup \{I\}$, that is, the inhibitory neuron is active at every fixed point since $b_i = \theta > 0$ and $b_I = 0$. 
For this reason, it is sufficient to focus on $\esupp(x^*)$ for E-I TLNs.  Moreover, the set $\{I\}$ cannot be a support since the external input satisfies $b_i=\theta > 0$ and $b_I = 0$.

Finally, we denote by $\EIFP$ the collection of all possible supports of fixed points in an $\EITLN$:
\begin{equation*}
    \EIFP \;\od \;\bigl\{\, \sigma\cup\{I\} \mid
    \sigma\cup\{I\} = \supp(x^*) \text{ for some fixed point } x^* 
    \text{ of the }  \EITLN \bigr\};
\end{equation*}
and $\EIFPe$, the collection of all possible e-supports of fixed points
\begin{equation*}
    \EIFPe \;\od \;
    \bigl\{\, \sigma \mid 
    \sigma = \esupp(x^*) \text{ for some fixed point } x^* 
    \text{ of the }  \EITLN \bigr\}.
\end{equation*}
The tools introduced above turn out to be very useful for analyzing the full fixed point structure.  
Indeed, consider a fixed point $x^*$ with $\esupp(x^*) = \sigma$.  
Then $x_i^* = y_i(x^*) > 0$ for all $i \in \sigma \cup \{I\}$, and $x_k^* = 0$ with $y_k(x^*) \le 0$ for all $k \in [n]\setminus\sigma$.  
Together with the fact that $y_I(x^*) > 0$, it follows that $x^*$ must lie in the chamber $R^y_{\sigma}$ and it is the fixed point of the corresponding linear system $L_\sigma$. 

\paragraph{Translation from CTLNs.}
The connection between CTLNs and E-I TLNs comes from the observation in~\cite{curto2025graphical} that, under the separation-of-timescales assumption $\tau_I \ll \tau_E = 1$, an E-I TLN effectively reduces to a CTLN. Specifically, an E-I TLN with parameters $(a,c)$ corresponds to a CTLN with parameters $(\delta,\varepsilon)$ via
\begin{equation*}
\delta = c-1, \qquad \varepsilon = a-c+1,
\end{equation*}
and the two models share exactly the same collection of fixed-point e-supports, as illustrated in Figure~\ref{fig:translation from CTLN}B.

Moreover, for any E-I TLN, the collection of fixed-point supports and e-supports is independent of $\tau_I$ and $\theta$. In particular,
\begin{equation}
    \EIFPe = \FPe(W) = \FPe(G,a,c),
\end{equation}
where the last equality follows from the fact that the connectivity matrix $W$ is determined by the underlying graph $G$ and the parameters $a,c$.
The proof is straightforward.
 Let $x^*$ be a fixed point of an $\EITLN(W,\theta,\tau_I)$.
 Then by definition, $x^*$ satisfies $x_i^* = [y_i^*]_+$. 
 Now consider any other ${(W,\theta',\tau'_I)}$ and define $x' = \theta' x^* / \theta$. It is 
 straightforward to verify that $x'_i = [y'_i]_+$, so $x'$ is the fixed point of the new E-I TLN.
 Moreover, since $\theta'/\theta > 0$, the signs of $x'_i$ and $x^*_i$ coincide, implying $\esupp (x') = \esupp (x^*)$.

 Combining these observations, any structural result about supports established for CTLNs in the region $\delta>0$ and $\varepsilon>0$ carries over directly to E-I TLNs in the parameter region $1<c<a+1$ with $a>0$. However, although the existence of fixed-point e-supports is independent of $\tau_I$, their stability is not in general, since the eigenvalues of the Jacobian at a fixed point depend on $\tau_I$.

For example, any singleton fixed point in a CTLN is always stable. Indeed, the reduced system on a singleton $\{k\}$ is given by
\[
\frac{dx_k}{dt} = -x_k + [\theta]_+.
\]
Since this one-dimensional system has the unique eigenvalue $-1$, the singleton fixed point is always stable. In contrast, the singleton fixed point in an E-I TLN is stable for $c < 1 + \frac{1}{\tau_I}$ and unstable for $c > 1 + \frac{1}{\tau_I}$ (see singleton E-I TLN). Thus, when importing results from CTLNs into the E-I setting, all stability conditions must be reconsidered. We next translate several key CTLN definitions and results that serve as foundational tools for our analysis of E-I TLNs.

In almost all previous CTLN studies, TLNs are required to be nondegenerate (which is true for most parameters). It is set to rule out line attractors, where a given support could give rise to infinite fixed points, and prevents the pathological case where fixed points from two distinct subsystems coincide. For nondegenerate TLNs, Morrison et al.\ have proved that the total number of fixed points is always odd~\cite{morrison2024diversity}. Due to this reason, our main results will also rely on nondegeneracy of E-I TLNs. The definition mirrors the notion of nondegenerate TLNs introduced in~\cite{morrison2024diversity}.
\begin{definition}[Nondegenerate E-I TLN]\label{def:nondegenerate}
Let $(W,b,\tau_I)$ be an E-I TLN on $n$ excitatory nodes and one inhibitory node. We say that it is nondegenerate if the following conditions hold:
\begin{itemize}
    \item $\det(I - W_{\sigma \cup \{I\}}) \neq 0$ for each $\sigma \subseteq [n]$;
    \item for each $\sigma \subseteq [n]$ such that $b_i > 0$ for all $i\in\sigma$, 
    the corresponding Cramer's determinant is nonzero:
    $
    \det\big((I-W_{\sigma \cup \{I\}})_i; b_{\sigma \cup \{I\}}\big) \neq 0;
    $
    \item $b_i > 0$ for at least one $i \in [n]$.
\end{itemize}
\end{definition}
Here we use the notation $A_\sigma$ and $b_\sigma$ to denote a submatrix of $A$ and a subvector of $b$ obtained by restricting to the indices in $\sigma$. We also use the notation $(A_i; b)$ to denote the matrix obtained from $A$ by replacing its $i$-th column with the vector $b$, as in Cramer's rule. 
The first condition, $\det(I-W_{\sigma\cup\{I\}})\neq 0$, ensures that all the linear subsystems of the E-I TLNs with $b_I \ge 0$ are nondegenerate. Hence there is only one fixed point for each linear subsystem. 
The second condition guarantees that the fixed points of two adjacent linear subsystems do not coincide on the common boundary of their respective chambers in state space.  
Finally, the third condition ensures that the origin is not a fixed point of the TLN.

Thus, for any nondegenerate TLN and a given subset $\sigma \subseteq [n]$, there is at most one fixed point with support $\sigma$. The proof for CTLNs has been given by Morrison et al.~\cite{morrison2024diversity} under the assumptions of competitive TLN ($W_{ij} \leq 0, W_{ii} = 0$, for $i,j \in  [n]$) and in the CTLN parameter range ($\delta>0$ and $\varepsilon>0$).  Here, we extend this uniqueness result and stability condition to E-I TLNs for the entire positive $(a,c)$ parameter plane. In addition, it is straightforward to show that, for any nondegenerate E-I TLN, the support of a fixed point must include the inhibitory node. Consequently, Lemma~\ref{lem:unique fp} concerns only supports of the form $\sigma\cup\{I\}$, where $\sigma\subseteq[n]$.
\begin{restatable}{lemma}{Lemmanondegenerate}\label{lem:unique fp}
    Let $(W,b,\tau_I)$ be a nondegenerate E-I TLN on $n$ excitatory nodes and one inhibitory node. For each $\sigma \subseteq [n]$, the E-I TLN has at most one fixed point with $\supp (x^*) = \sigma \cup \{I\}$. If such a fixed point exists, it is stable if and only if $T_{\sigma \cup \{I\}}^{-1}(-I + W_{\sigma \cup \{I\}})$ is a stable matrix (i.e., all eigenvalues have negative real part), where $T=diag(1, \dots, 1, \tau_I)$ is the $(n+1)\times(n+1)$ diagonal timescale matrix.
\end{restatable}

\begin{proof}
Suppose $\sigma \subseteq [n]$. Without loss of generality, we can reorder the nodes so that $\sigma = \{1,\dots,|\sigma|\}$ are the active nodes and  $\bar\sigma = [n]\setminus\sigma = \{|\sigma|+1,\dots,n\}$ are the silent nodes. Moreover, the inhibitory neuron is the $(n+1)$-st neuron, which is active for any support since $b_I \ge 0$ and $b_i=\theta>0$. 
Let $x^*$ be a fixed point with $\supp(x^*) = \sigma \cup \{I\}$. Thus, $x^* \in R_{\sigma}$, where the E-I TLN reduces to a linear system $L_\sigma$:
{\renewcommand{\arraystretch}{1.5}
\begin{equation*}
    T \frac{dx}{dt} = A x + b
    \quad\text{for}\quad 
    A = \left[
    \begin{array}{cc|c}
    -I_{|\sigma|} + W_{\sigma\sigma} & W_{\sigma \bar\sigma}            & -\mathbf{1}_{|\sigma|} \\
    0                                & -I_{|\bar\sigma|}                & 0 \\
    \hline
    c \mathbf{1}_{|\sigma|}^\top     & c \mathbf{1}_{|\bar\sigma|}^\top & -1
    \end{array}
    \right], \quad
    b = \left[
    \begin{array}{c}
    b_{\sigma} \\ 0 \\ \hline b_I
    \end{array}
    \right],\quad
    T = \left[
    \begin{array}{cc|c}
    I_{|\sigma|}          & 0                & 0\\
    0                     & I_{|\bar\sigma|} & 0\\
    \hline 0              & 0                & \tau_I
    \end{array}
    \right].
\end{equation*}
}
Here $W_{\sigma\sigma}$ is the submatrix of $W$ restricted to $\sigma$ and $W_{\sigma\bar\sigma}$ is the submatrix of $W$ with rows indexed by $\sigma$ and columns indexed by $\bar\sigma$.  $I_{|\sigma|}$ and $I_{|\bar\sigma|}$ are identity matrices of sizes
$|\sigma|$ and $|\bar\sigma|$. $\mathbf{1}_k$ denotes the $k$-dimensional all-ones column vector.  For an E-I TLN, $b_\sigma=\theta\,\mathbf{1}_{|\sigma|}$, $b_I=0$, and $\tau_E=1$ for all excitatory nodes (hence the first $n$ diagonal entries of $T$ are equal to 1).

For the fixed point $x^*$ of the E-I TLN with e-support $\sigma$, it must also be a fixed point of $L_\sigma$. Since this E-I TLN is nondegenerate (then $\det A \neq 0$ from Definition~\ref{def:nondegenerate}), $L_\sigma$ has a unique fixed point $x'$, which is given by 
$$
x' = (-A)^{-1} \begin{bmatrix} b_{\sigma} \\ 0 \\ b_I \end{bmatrix},
$$
and $x'$ becomes a fixed point of the original E-I TLN if and only if $x' \in R_\sigma$ and $y_i(x') \leq 0$ for $i \in \bar\sigma$. Therefore, there is at most one fixed point with e-support $\sigma$ since there is only one fixed point candidate.

Moreover, if such a fixed point $x^*$ exists, its stability is determined by the Jacobian of $L_\sigma$ at $x^*$.
Since $L_\sigma$ is linear, the Jacobian is constant and equals $T^{-1}A$, i.e.,
{\renewcommand{\arraystretch}{1.5}
\begin{equation*}
T^{-1}A = \left[
\begin{array}{cc|c}
-I_{|\sigma|} + W_{\sigma\sigma} & W_{\sigma \bar\sigma} & -\mathbf{1}_{|\sigma|}\\
0 & -I_{|\bar\sigma|} & 0\\
\hline
\frac{c}{\tau_I} \mathbf{1}_{|\sigma|}^\top 
& \frac{c}{\tau_I} \mathbf{1}_{|\bar\sigma|}^\top 
& - \frac{1}{\tau_I}
\end{array}
\right]
\end{equation*}
}
Now permute coordinates to the order $(1,\dots, |\sigma|,I, |\sigma|+1, \dots, n )$ so that the first $|\sigma|+1$ neurons are active and the remaining neurons are silent.  Under this permutation, the Jacobian is similar to $T^{-1}A$ , and it has block upper triangular form:
{\renewcommand{\arraystretch}{1.5}
\begin{equation*}
    T^{-1}A \sim 
    \left[
    \begin{array}{ccc}
    - I_{|\sigma|} + W_{\sigma\sigma}          & -\mathbf{1}_{|\sigma|}  & W_{\sigma \bar\sigma} \\
    \frac{c}{\tau_I}\mathbf{1}_{|\sigma|}^\top & -\frac{1}{\tau_I}       & \frac{c}{\tau_I}\mathbf{1}_{|\bar\sigma|}^\top \\
     0 & 0 & -I_{|\bar\sigma|}
    \end{array}
    \right]
    =
    \left[
    \begin{array}{cc}
    B & * \\
    0 & -I_{|\bar\sigma|}
    \end{array}
    \right].
    \end{equation*}
}
where $B = T_{\sigma\cup\{I\}}^{-1}(-I + W_{\sigma\cup\{I\}})$. Since similarity preserves eigenvalues, we may compute the spectrum from this block form.  The characteristic polynomial factors as
$$
\det(\lambda I - T^{-1}A)
=
\det(\lambda I - B)\,\det(\lambda I + I_{|\bar\sigma|}),
$$
and therefore
$$
\spec(T^{-1}A)
=
\spec(B)\cup\{ \underbrace{-1, \dots, -1}_{|\bar\sigma|\text{ times}} \}.
$$
In particular, the eigenvalues associated with the silent coordinates are all equal to $-1$, and the remaining eigenvalues are
exactly those of $B$.  Therefore, $x^*$ is stable if and only if $B$ is a stable matrix.
\end{proof}

Lemma~\ref{lem:unique fp} establishes that, for a nondegenerate E-I TLN, each support corresponds to at most one fixed point. Building on the nondegenerate CTLN framework, a series of graph rules has been developed, including domination theory and uniform in-degree results~\cite{curto2023graph}.
In this paper, we also restrict our attention to nondegenerate E–I TLNs. Within this setting, we characterize all possible supports for $n$-paths and $n$-cycles, with the results and proofs given in the next section.

\subsection{Proofs of the graph rules}
\label{sec:proof_support}

In this section, we provide proofs of the graph rules and establish the supporting lemmas used to characterize fixed-point e-supports. Our arguments rely primarily on domination theory and direct fixed-point analysis. These graph rules underlie the characterizations for nondegenerate E-I TLNs on paths and cycles given in Theorems~\ref{thm:path_supports} and~\ref{thm:cycle_supports}.

Recall that a point $x^*$ is a fixed point if $x^*_i = [y^*_i]_+$, $i \in [n] \cup \{I\}$. Then, a subset $\sigma \cup \{I\}$ is the support of $x^*$ if  $x^*_i>0$ for $i \in \sigma \cup \{I\}$ and $x^*_j=0$ for $j \in [n] \setminus \sigma$. Thus, the fixed point conditions naturally split into on- and off-neuron conditions:
\begin{enumerate}
    \item  \emph{On-neuron condition:} for each $i \in \sigma \cup \{I\}$, $x^*_i = y_i(x^*)>0$. This condition ensures that the neurons in $\sigma$ (and the inhibitory neuron $I$) are active and hence is referred to as the on-neuron condition. This condition has two parts. First, $x^*$ must be the solution of
    \begin{equation}\label{eq:on_neuron}
        (I-W_{\sigma \cup \{I\}})x_{\sigma \cup \{I\}} = b_{\sigma \cup \{I\}}.
    \end{equation}
    Since we only consider nondegenerate E-I TLNs, $I-W_{\sigma\cup\{I\}}$ is
    invertible. Hence, \Eq{eq:on_neuron} has a unique solution. Second, all
    coordinates of this solution must be positive, so that
    $x^*_i = y_i(x^*)>0$ for all $i\in\sigma\cup\{I\}$. The on-neuron condition is exactly the same as the fixed point condition for the full-support fixed point of the restricted network on $G|_\sigma$. Thus, satisfying the on-neuron condition is equivalent to the existence of a full-support fixed point in the restricted network on $G|_\sigma$, i.e., $\sigma \in \FPe(G|_\sigma,a,c)$.
    \item \emph{Off-neuron condition:} for each $i \in [n]\setminus \sigma$, $y_i(x^*) \leq 0$.
    This condition ensures that all neurons outside $\sigma$ remain silent. 
    If it holds, then the fixed point $x^*|_{\sigma\cup\{I\}}$ of the restricted E-I TLN on $G|_\sigma$ extends to the full network by keeping the same coordinates on $\sigma\cup\{I\}$ and setting all coordinates outside $\sigma$ to zero. 
    Thus,
    \[
    \sigma\in\FPe(G|_\sigma,a,c)
    \Rightarrow
    \sigma\in\FPe(G,a,c).
    \]
\end{enumerate}
The on- and off-neuron conditions together characterize when a subset $\sigma \subseteq [n]$ can serve as a fixed point e-support. In what follows, we introduce six lemmas for checking candidate supports in E-I TLNs on paths and cycles. We first restate the domination theory from~\cite{curto2025graphical} as Lemma~\ref{lem:domination}, which constrains the fixed point structure in the moderate inhibition regime. Beyond the moderate inhibition regime, we introduce weak domination and four additional lemmas. 
The former excludes certain subsets from being supports in the weak inhibition regime (Lemma~\ref{lem:weakdomination}) while the latter (Lemmas~\ref{lem:uniform_indegree_on},~\ref{lem:uniform_indegree_off},~\ref{lem:paths} and~\ref{lem:off condition}) are used to verify the on- and off-neuron conditions for candidate supports in paths and cycles, as summarized in Table~\ref{tab:lemmas}.

Throughout the proof, we frequently use two specific types of subgraphs: subgraphs with uniform in-degree and independent sets. 
We recall their definitions below.
\begin{definition}[Uniform in-degree]
Let $G$ be a directed graph on $n$ vertices and let $\sigma\subseteq[n]$. 
We say that the induced subgraph $G|_\sigma$ has \emph{uniform in-degree} $d$ if every node in $\sigma$ has $d$ incoming edges within $G|_\sigma$.
\end{definition}

\begin{definition}[Independent set]
Let $G$ be a directed graph on $n$ vertices and let $\sigma\subseteq[n]$. 
We say that $G|_\sigma$ is an \emph{independent set} if there are no directed edges within $G|_\sigma$. 
\end{definition}

\begin{table}[H]
\footnotesize
\renewcommand{\arraystretch}{1.3}
\caption{Lemmas for verifying on-neuron and off-neuron conditions}\label{tab:lemmas}
\begin{center}
  \begin{tabular}{c|c|c} \hline
      &  directed subgraph & result   \\ \hline
   Lemma~\ref{lem:uniform_indegree_on} &  subgraphs with uniform in-degree & necessary and sufficient on-neuron condition  \\ \hline
   Lemma~\ref{lem:uniform_indegree_off} & subgraphs with uniform in-degree & necessary and sufficient  off-neuron condition  \\ \hline
   Lemma~\ref{lem:paths} &  disjoint unions of one or more paths & sufficient on-neuron condition  \\ \hline
   Lemma~\ref{lem:off condition} &  arbitrary subgraphs & sufficient off-neuron condition \\ \hline
  \end{tabular}
\end{center}
\end{table}
It turns out that these six lemmas are particularly useful for analyzing E-I TLNs on paths and cycles.
Graphical domination (Lemma~\ref{lem:domination}) and weak domination (Lemma~\ref{lem:weakdomination}) allow us to exclude certain subsets as supports in the moderate and weak inhibition regimes, respectively. 
Lemmas~\ref{lem:uniform_indegree_on} and~\ref{lem:uniform_indegree_off} apply to subgraphs with uniform in-degree, such as independent sets and cycles.
Moreover, since any proper induced subgraph of a path or a cycle is a disjoint union of one or more paths, Lemmas~\ref{lem:paths} and~\ref{lem:off condition} are used to show that all such subsets survive as supports in the strong inhibition regime.

\subsubsection{Lemmas for domination and weak domination}

Lemma~\ref{lem:domination} is from the theory of graphical domination developed in~\cite{curto2025graphical}, which provides strong combinatorial constraints on possible supports. This notion captures a specific domination relation between vertices in the underlying graph $G$. 
It was first introduced in~\cite{curto2019fixed} and is defined as follows.
\begin{definition}[Graphical domination]\label{def:domination}
    Let $j,k \in [n]$ be vertices of $G$. We say that $k$ \emph{graphically dominates} $j$ in $G$ if the following two conditions hold:
    \begin{enumerate}
    \item For each $i \in [n]\setminus\{j,k\}$, if $i \to j$ then $i \to k$.
    \item $j \to k$ and $k \not\to j$.
    \end{enumerate}
\end{definition}

If such a pair $(j,k)$ exists, then $j$ is called a dominated node (or dominated vertex) of $G$.
For E-I TLNs with arbitrary $\tau_I$ and in the moderate inhibition regime ($1< c<a+1$ with $a>0$), dominated nodes can be removed to obtain a smaller subnetwork that preserves the fixed point structure. 
This fact is fully established in~\cite{curto2025graphical}.
Below we restate Lemma 3.8 of~\cite{curto2025graphical} in the setting of E–I TLNs with parameters $a$ and $c$.
\begin{lemma}[Lemma~3.8, \cite{curto2025graphical}]\label{lem:domination}
Suppose $j$ is a dominated node in a directed graph $G$. Then for nondegenerate E-I TLNs in the moderate inhibition regime ($1 < c < a+1$ with $a>0$), the fixed point supports of the network constructed from $G$ satisfy
\begin{equation*}
    \FPe(G,a,c) = \FPe(G|_{[n]\setminus \{j\}},a,c).
\end{equation*}
\end{lemma}
The proof of this lemma is given in~\cite{curto2025graphical}. 
Lemma~\ref{lem:domination} is particularly powerful in the classical CTLN regime $\delta>0, \varepsilon>0$, which corresponds to the moderate inhibition regime in E-I TLNs.
When inhibition becomes even stronger ($c>a+1$) or weaker ($c<1$), corresponding to the strong or weak inhibition regime, the domination relation no longer yields the same consequences. For example, consider the $2$-node path $1 \to 2$. In the strong inhibition regime, we obtain $\FPe(G,a,c) = \{ \{2\},  \{1\},\{1,2\}\}$ and in the weak inhibition regime, we obtain $\FPe(G,a,c) = \{\{1,2\}\}$. These examples show that a dominated node cannot, in general, be removed in the strong or weak inhibition regime.
 
However, in the weak inhibition regime, there is still a graphical relation that constrains the fixed point structure. 
We call this relation \emph{weak graphical domination} since it is obtained by dropping the requirement 
$j \to k$ and $k \not\to j$ from Definition~\ref{def:domination}.

\begin{definition}[Weak graphical domination]
    Let $j,k \in [n]$ be vertices of $G$. We say that $k$ \emph{weakly graphically dominates} $j$ in $G$ if for each $i \in [n]\setminus\{j,k\}$, $i \to j $ implies $i \to k$.
\end{definition}
In the E-I TLN setting, $k$ weakly graphically dominates $j$ in $G$ leads to $W_{ji}\leq W_{ki}$ for any node $i\neq j,k $ in $G$. 
Unlike graphical domination, weak graphical  domination imposes no condition on the edge between $j$ and $k$. 
Thus, it can be symmetric: if $j$ and $k$ receive the same incoming edges from the rest of $G$, then they weakly graphically dominate each other. 
For instance, in the $n$-path $1 \to 2 \to \dots \to n$, node $i$ weakly dominates node $i+1$ for $i \in [n-1]$, and node $1$ is weakly dominated by any other node.

We now show how weak domination constrains the fixed point e-supports in the weak inhibition regime $0<c<1$ with $a>0$.

\begin{lemma}\label{lem:weakdomination}
Let $j,k \in [n]$ be vertices of $G$ and suppose $k$ weakly graphically dominates $j$ in $G$.
Then for any nondegenerate E-I TLN in the weak inhibition regime ($0<c<1$ and $a>0$), if $\sigma \in \FPe(G,a,c)$, we have
\begin{equation*}
     \quad j \in \sigma \Rightarrow k \in \sigma.
\end{equation*}
\end{lemma}
\begin{proof}
From the definition of $y_j$ and $y_k$ in the E-I TLN:
\begin{align*}
  y_k &= W_{kk} x_k + W_{kj} x_j
  + \sum_{i \in [n] \setminus \{j,k\}} W_{ki} x_i + W_{kI}x_I + b_k,\\
  y_j &= W_{jj} x_j + W_{jk} x_k
  + \sum_{i \in [n] \setminus \{j,k\}} W_{ji} x_i + W_{jI}x_I + b_j.
\end{align*}
Since $k$ weakly graphically dominates $j$ in $G$, we have $W_{ji}\leq W_{ki}$. Moreover, in the E-I TLN setting, $W_{jI}=W_{kI}$ and $b_j=b_k=\theta$. Thus, at any $x\in \mathbb{R}^{n+1}_{\ge 0}$, it follows that
\begin{equation*}
    y_k \ge  y_j + (W_{kk}-W_{jk}) x_k+ (W_{kj}-W_{jj}) x_j.
\end{equation*}
Now suppose, for contradiction, that $x^*\in \mathbb{R}^{n+1}_{\ge 0}$ is a fixed point with e-support $\sigma$, where $j \in \sigma$ and $k \notin \sigma$. Equivalently, $x^*_j = y^*_j>0$ and $x^*_k = [y^*_k]_+ = 0$, so $y_k^* \le 0$. But applying the last inequality at $x=x^*$, with $x_k^*=0$, $x_j^*=y_j^*$, $W_{jj}=c$, and $W_{kj}\geq 0$, gives
\begin{equation*}
    y^*_k \ge  y^*_j+ (W_{kj}- c ) y^*_j \ge  (1- c) x^*_j>0,
\end{equation*}
since $0<c<1$, contradicting $y_k^*\leq 0$. Hence, for any $\sigma$ such that $k\notin\sigma$, $j\in\sigma$, $\sigma \not\in \FPe(G,a,c)$.  
Therefore $k\notin\sigma$ forces $j\notin\sigma$, or equivalently, $j\in\sigma$ implies $k\in\sigma$.
\end{proof}

\subsubsection{Lemmas for on- and off-neuron conditions}
Lemmas~\ref{lem:uniform_indegree_on} and ~\ref{lem:uniform_indegree_off} below originate from the  uniform in-degree theory for CTLNs~\cite{curto2019fixed}, which establishes the existence of certain fixed point supports for CTLNs in the legal range ($\delta,\varepsilon >0, \varepsilon<\delta/(\delta+1)$).  
Here, we derive analogous on- and off-neuron conditions for nondegenerate E--I TLNs throughout the positive $(a,c)$ parameter space, beyond the moderate inhibition regime (which corresponds to the CTLN legal range).

\begin{lemma}[On-neuron condition for uniform in-degree subgraphs]\label{lem:uniform_indegree_on}
Let $G$ be a directed graph. Suppose the induced subgraph $G|_\sigma$
has  uniform in-degree $d$. Then for any nondegenerate E-I TLN restricted on $G|_\sigma$ with parameters $a$, $c$,
$$
\sigma \in \FPe(G|_\sigma,a,c)
\iff
(|\sigma| - 1)c - d a + 1 > 0.
$$
In this case, the corresponding fixed point is
\[
x^*_{\sigma \cup \{I\}}
=
\alpha
\begin{pmatrix}
\mathbf{1}_\sigma \\
|\sigma|c
\end{pmatrix},
\qquad
\alpha
=
\frac{\theta}{(|\sigma|-1)c - da + 1},
\]
where $|\sigma|$ denotes the cardinality of $\sigma$, and
$\mathbf{1}_\sigma$ is the all-ones vector of size $|\sigma|$.
\end{lemma}

\begin{proof}
Let $W_{\sigma \cup \{I\}}$ be the connectivity matrix for the restricted E-I TLN on $G|_\sigma$. Then, determining whether $\sigma\cup\{I\}$ supports a fixed point is equivalent to
determining whether \Eq{eq:on_neuron},
\[
    (I_{|\sigma|+1}-W_{\sigma \cup \{I\}})x_{\sigma \cup \{I\}}
    = b_{\sigma \cup \{I\}}
\]
admits a strictly positive solution. 
Since the induced subgraph $G|_\sigma$
has uniform in-degree $d$, the excitatory submatrix $I_{|\sigma|}-W_\sigma$ has uniform row sum
$1-c-da$. Moreover, in the E-I TLN setting,
$b_{\sigma\cup\{I\}}=(\theta\mathbf{1}_\sigma,0)^T$. This motivates us to look for a solution in which all excitatory firing rates are equal. From the inhibitory equation in~\Eq{eq:on_neuron}, we have $x^*_I=c\sum_{i\in\sigma}x^*_i$ at the fixed point. Thus, we consider a candidate solution 
\[x_{\sigma \cup \{I\}}^*=\alpha(1,1,\dots,1,|\sigma| c)^T.\]

Since $W_{iI}=-1$ and $W_{Ii}=c$ for all $i\in\sigma$, and since the rows
of the submatrix $I_{|\sigma|}-W_\sigma$ have sum $1-c-da$, substituting $x_{\sigma \cup \{I\}}^*$ into \Eq{eq:on_neuron} gives
\begin{align*}
    (I-W_{\sigma \cup \{I\}})x^*_{\sigma \cup \{I\}} &= \alpha \begin{pmatrix}
        (1-c-da+|\sigma|c)\mathbf{1}_\sigma \\
        c|\sigma| - c|\sigma|
    \end{pmatrix} \\
    &= \alpha(1-c-da+|\sigma|c)
    \begin{pmatrix}
        \mathbf{1}_\sigma \\
        0
    \end{pmatrix} = \begin{pmatrix}
        \theta \mathbf{1}_\sigma \\
        0
    \end{pmatrix}
\end{align*}
Hence, $\alpha=\theta/((|\sigma|-1)c-da+1)$ gives the desired result. 
By nondegeneracy, this solution is unique. Thus, $\sigma\cup\{I\}$ supports
a fixed point in the restricted E-I TLN on $G|_\sigma$ if and only if this unique solution is strictly positive. Since
$\theta>0$ and $c>0$, this is equivalent to
$(|\sigma| - 1)c - da + 1 > 0$.
Therefore, $\sigma \in \FPe(G|_\sigma,a,c)
    \iff
    (|\sigma| - 1)c - da + 1 > 0$.
\end{proof}

Lemma~\ref{lem:uniform_indegree_on} also implies that if the induced subgraph $G|_\sigma$ has  uniform in-degree, then the fixed point supported on $\sigma \cup \{I\}$ has the same value for all excitatory entries, even when $G|_\sigma$ is not symmetric. We next derive the off-neuron conditions for such subgraphs.

\begin{lemma}[Off-neuron condition for uniform in-degree subgraphs]\label{lem:uniform_indegree_off}
Let $G$ be a directed graph on $n$ vertices, and let $G|_\sigma$ be an induced proper subgraph with uniform in-degree $d$. For each $k\in [n]\setminus\sigma$, define
\[
d_k:=\big|\{\, i\in\sigma \mid i\to k \text{ in } G \,\}\big|
\]
to be the number of edges from $\sigma$ to $k$.
If $\sigma \in \FPe(G|_{\sigma},a,c)$, then
\[
\sigma \in \FPe(G,a,c) \iff c \ge (d_k-d)a+1, \forall k\in [n] \setminus \sigma. 
\]
\end{lemma}
\begin{proof}
    Since $G|_\sigma$ has uniform in-degree $d$ and $\sigma \in \FPe(G|_{\sigma},a,c)$, Lemma~\ref{lem:uniform_indegree_on} gives the fixed point supported on $\sigma\cup\{I\}$ in the restricted E-I TLN on $G|_\sigma$:
    \[
    x^*_{\sigma \cup \{I\}}
    =
    \alpha
    \begin{pmatrix}
    \mathbf{1}_\sigma \\
    |\sigma|c
    \end{pmatrix},
    \qquad
    \alpha
    =
    \frac{\theta}{(|\sigma|-1)c - da + 1}>0,
    \]
    We now determine when this fixed point persists in the full E--I TLN on $G$. 
    Define a candidate equilibrium
    $$
    x'_i = x_i^* \text{ for } i \in \sigma, 
    \qquad x'_i = 0 \text{ for }  i \notin \sigma, 
    \qquad x'_I = x_I^* .
    $$   
    For any node $k \in [n]\setminus \sigma$, the total input to $k$ at $x'$ is
    \begin{align*}
    y_k(x') &= \sum_{i\in \sigma}W_{ki}x'_i + \sum_{i \notin \sigma}W_{ki}x'_i - x'_I + \theta \\
    & =a d_k  \alpha - |\sigma|c\alpha + \theta \\
    & = (d_k a -|\sigma|c + (|\sigma|-1)c - da + 1)\alpha \\
    & = ((d_k-d) a - c  + 1)\alpha.
    \end{align*}
    Thus, $y_k(x') \le 0$ for all $k \in [n]\setminus \sigma$ is equivalent to $(d_k-d) a - c  + 1 \le 0$ for all $k \in [n]\setminus \sigma$. Thus $\sigma \in \FPe(G,a,c) \iff c \ge (d_k-d)a+1, \forall k\in [n] \setminus \sigma$.
\end{proof}

We now consider the subgraphs that can arise in paths and cycles. 
A useful feature of $n$-paths and $n$-cycles is that every proper induced subgraph is a disjoint union of one or more directed paths, where a singleton is regarded as a $1$-path. 
This observation motivates Lemma~\ref{lem:paths}, which establishes the existence of full-support fixed points for all subnetworks of E--I TLNs on paths and cycles.

\begin{lemma}[On-neuron condition for union of paths]\label{lem:paths}
    Let $G$ be a directed graph, and let the subgraph $G|_{\sigma}$ be a disjoint union of one or more directed paths. For any nondegenerate E-I TLN, if $(a,c)$ lies in the strong or weak inhibition regime, i.e., if $c>a+1$ or $0<c<1$ (with $a>0$), then $\sigma \in \FPe(G|_{\sigma},a,c)$.
\end{lemma}

\begin{proof}
    Recall that $\sigma\cup\{I\}$ supports a fixed point in the restricted E-I TLN on $G|_\sigma$ if and only if \Eq{eq:on_neuron},
    \[
    (I-W_{\sigma \cup \{I\}})x_{\sigma \cup \{I\}} = b_{\sigma \cup \{I\}}
    \]
    admits a strictly positive solution. By nondegeneracy, $I-W_{\sigma\cup\{I\}}$ is invertible, so the solution is unique and is given by $x_{\sigma \cup \{I\}}^* = (I-W_{\sigma \cup \{I\}})^{-1}b_{\sigma \cup \{I\}}$. However, $x_{\sigma \cup \{I\}}^*$ may not be positive.
    Next, we show that in the strong and weak inhibition regimes, $x_{\sigma \cup \{I\}}^*$ has strictly positive coordinates and hence survives as a fixed point of this subnetwork when $G|_{\sigma}$ is a disjoint union of one or more paths.

    First, suppose that the induced subgraph $G|_\sigma$ is a directed path. Without
    loss of generality, we may write $G|_\sigma: 1\to 2\to \cdots \to m$. Then we have
    \[
    I - W_{\sigma\cup\{I\}}=
    \begin{pmatrix}
    1-c & 0 & 0 & \cdots & 0 & 1\\
    -a & 1-c & 0 & \cdots & 0 & 1\\
    0 & -a & 1-c & \cdots & 0 & 1\\
    \vdots & \vdots & \vdots & \ddots & \vdots & \vdots\\
    0 & 0 & 0 & \cdots & 1-c & 1\\
    -c & -c & -c & \cdots & -c & 1
    \end{pmatrix},
    \qquad
    b_{\sigma\cup\{I\}}=
    \begin{pmatrix}
    \theta\\
    \theta\\
    \vdots\\
    \theta\\
    0
    \end{pmatrix}.
    \]
Therefore, \Eq{eq:on_neuron} gives
\begin{align}\label{eq:path_first}
    (1-c)x^*_{1} + x^*_I &= \theta,\\
    \label{eq:path_other}
    (1-c)x^*_{i} - a x^*_{i-1}  + x^*_I &= \theta,
    \quad i = 2,\dots,m,\\
    \label{eq:path_inh}
    - c\sum_{i\in[m]} x^*_{i} + x^*_I  &= 0.
\end{align}
    If $m=1$, then only \Eq{eq:path_first} and \Eq{eq:path_inh} are needed. 
    Since $c\neq 1$ in the strong and weak inhibition regimes, we can solve for $x_1^*$ as
    \[
    x_1^* = \frac{-x^*_I + \theta}{1-c}.
    \]
    For convenience, define $p := (-x_I^*+\theta)/(1-c)$ and $q:= a/(1-c)$. Then $x_1^*=p$, and
    \[
    x_2^* = \frac{a x^*_1 + \theta -x^*_I}{1-c} = \frac{a}{1-c}p+p = p(1+q).
    \]
    The same recursion gives the coordinates $x_{i}^*$ in terms of $p$ and $q$:
    \[
    x^*_i = p(1+q+\cdots+q^{i-1}),\quad i = 1,2,\dots,m.
    \]
    In either inhibition regime, $c>a+1$ (strong) or $0<c<1$ (weak), we have $q=a/(1-c)>-1$. Therefore, if $-1< q < 0$, then $1+q+\cdots+q^{i-1} = (1-q^{i})/(1-q) >0$ for every $i\in[m]$ while if $q \ge 0$, $1+q+\cdots+q^{i-1}>0$ is obvious. Hence, all excitatory coordinates $x^*_i$ have the same sign as $p$ and so does $x^*_I$ since $x_I^* = \sum_{i\in [m]}c x_i^*$ from \Eq{eq:on_neuron}.
    
    Now substituting $x_I^* = \sum_{i\in [m]}c x_i^*$ into \Eq{eq:path_first} gives
    \begin{equation*}
     x_1^* + \sum_{i=2}^m c x_i^* = \theta.
    \end{equation*}
    The left-hand side has the same sign as $p$, since all $x_i^*$ have the same sign as $p$ and $c>0$. Since $\theta>0$, we must have $p>0$. This implies $x_i^*>0$ for all $i\in[m]$ and $x_I^*>0$.  Therefore, when $G|_\sigma$ is a directed path, in both strong and weak inhibition regimes, $x^*$ is a fixed point supported on $\sigma \cup \{I\}$, and hence
    $\sigma \in \FPe(G|_{\sigma},a,c)$.

    Next, we consider the case when $G|_\sigma$ is a disjoint union of paths. Fix an arbitrary path component, say $i_1 \to i_2 \to \dots \to i_m$. 
    Since there are no edges between distinct path components, the on-neuron equations \Eq{eq:on_neuron} along this path have the same recursive form as \Eq{eq:path_first} and \Eq{eq:path_other}:
     \begin{align*}
        (1-c)x^*_{i_1} + x^*_I &= \theta,\\
    (1-c)x^*_{i_k} - a x^*_{i_{k-1}}  + x^*_I &= \theta,
    \quad k = 2,3,\dots,m.
    \end{align*}
    The only difference is that the inhibitory node now depends on all path components: $x_I^*=c\sum_{j\in\sigma}x_j^*$. By again setting $p = (-x^*_I + \theta)/(1-c)$ and $q=a/(1-c)$, we can express the coordinates $x_{i_k}^*$ along this path in terms of $p$ and $q$, as before:
    \[
    x^*_{i_{k}} = p(1+q+\cdots+q^{k-1}),\quad k = 1,2,\dots,m.
    \]
    The same argument applies to every other path component. Thus, all excitatory coordinates $x_i^*$ for $i\in\sigma$ and the $x^*_I$ have the same sign as $p$.
    By the same argument as above, substituting $x_I^*=c\sum_{j\in\sigma}x_j^*$ into the first equation of any path component shows that $p>0$. Hence, $x_i^*>0$ for all $i\in\sigma$ and $x_I^*>0$.
    Therefore, in both inhibition regimes, when $G|_\sigma$ is a disjoint union of directed paths, $x^*$ is a fixed point supported on $\sigma\cup\{I\}$ in the restricted E-I TLN on $G|_\sigma$. Hence $\sigma\in\FPe(G|_\sigma,a,c)$.
\end{proof}

After verifying the on-neuron conditions for a subgraph $G|_\sigma$ of a path or cycle, we must check the off-neuron conditions to determine whether the corresponding fixed point in the restricted subnetwork survives in the full E-I TLN on $G$. 
The final lemma provides sufficient off-neuron conditions for arbitrary subgraphs.

\begin{lemma}[Off-neuron conditions for arbitrary subgraphs]
\label{lem:off condition}
Let $G$ be a directed graph on $n$ vertices, and let $G|_\sigma$ be an induced proper subgraph. For each $k\in [n]\setminus\sigma$, define
\[
d_k:=\big|\{\, i\in\sigma \mid i\to k \text{ in } G \,\}\big|
\]
to be the number of edges from $\sigma$ to $k$. If $\sigma \in \FPe(G|_{\sigma},a,c)$ and $c\ge d_k a+1$ for all $k\in [n]\setminus\sigma$, then
\[
\sigma \in \FPe(G,a,c).
\]
\end{lemma}
\begin{proof}
    Since $\sigma \in \FPe(G|_{\sigma},a,c)$, let $x^*$ be a fixed point with support $\sigma \cup \{I\}$ in the restricted E-I TLN on $G|_\sigma$.
    Choose $j\in\sigma$ such that $x_j^*=\max_{i\in\sigma}x_i^*$.  For this node $j$, the on-neuron condition gives
    $$
    x_j^* =y_j^* =  c x_j^* + \sum_{i\in\sigma,\, i\ne j} W_{ji}x_i^* - x_I^* + \theta>0 .
    $$
    Since all $W_{ji} \ge 0$ and $x_i^*>0$, we have
    $
    (1-c)x_j^* \ge -x_I^* + \theta
    $. 
    
    Now construct a candidate fixed point of the full network.  
    Define $x'$ by
    $$
    x'_i = x_i^*\text{ for } i \in \sigma, 
    \qquad x'_i = 0\text{ for } i \notin \sigma, 
    \qquad x'_I = x_I^* .
    $$
    We next verify that $x'$ satisfies the fixed point conditions for the full E-I TLN. For any node $k \in [n]\setminus \sigma$, the total input to $k$ at $x'$ is
    \begin{align*}
    y_k(x') &= \sum_{i\in \sigma}W_{ki}x'_i + \sum_{i \notin \sigma}W_{ki}x'_i - x'_I + \theta \\
    & \le d_k a x_j^* - x_I^* + \theta \\
    & \le (d_k a+1-c)x_j^* \le 0,
    \end{align*}
    where the third line follows from $-x_I^*+\theta \le (1-c)x_j^*$.  Thus, $[y_k(x')]_+ = 0 = x'_k$ for $k \in [n] \setminus \sigma$, and  $[y_i(x')]_+ = y_i^* = x_i^* = x'_i>0$ for $i \in \sigma\cup\{I\}$.  Hence $x'$ is a fixed point supported on  $\sigma \cup \{I\}$ in the E-I TLN on $G$, i.e.
    $
    \sigma \in \FPe(G,a,c).
    $
\end{proof}

\subsection{Proof of Theorem \ref{thm:cycle_2_modes}}
\label{sec:z-mode}

Here we show that the $z$-mode and mean mode decouple in the full-support chamber when the graph $G$ has both  uniform in-degree and  out-degree and prove Theorem~\ref{thm:cycle_2_modes}.
We recall the following definition.

\begin{definition}[Uniform in-degree and out-degree]
We say that $G$ has  uniform in-degree $d$ if every vertex receives exactly $d$ incoming edges in $G$, and that $G$ has  uniform out-degree $d$ if every vertex in $G$ sends exactly $d$ outgoing edges in $G$.
\end{definition}

Here, the notion of degree refers purely to the underlying directed graph $G$. Hence, the prescribed self-excitation in E-I TLNs is not counted.
Under these definitions, we obtain the following decoupling result.

\begin{restatable}{theorem}{thmfour}\label{thm:uniform degree}
Let $[n]\in \FPe(G,a,c)$.
Then in the chamber $R_{[n]}$:
\begin{itemize}
    \item If $G$ has  uniform in-degree, the $z$-mode is independent of the mean variables $(x_E,x_I)$.
    \item If $G$ has  uniform out-degree, the mean mode is independent of the difference variables $z$.
\end{itemize}
\end{restatable}
\begin{proof}
    First, we consider the case where $G$ has  uniform in-degree $d_{in}$.
    By the definition of $z$ and $x_E$, each $x_i$ can be decomposed into an average term  plus a fluctuation term depending only on $z$. Explicitly, there exist functions $f_i:\mathbb{R}^{\,n-1}\to\mathbb{R}$, $i\in[n]$, such that
    \begin{equation*}
        x_i = \frac{x_E}{n} + f_i(z).
    \end{equation*}
    Here $f_i(z)$ encodes the deviation of $x_i$ from the mean value, and satisfies $\sum_{i\in [n]} f_i(z)=0$.  
    Then, for each $z_j$, we compute
    \begin{align*}
        \frac{d z_j}{dt}  = & \frac{d x_{j+1}}{dt} - \frac{d x_{j}}{dt}  \\
         = & -x_{j+1} + \left[ \sum_{k=1}^{n} W_{j+1,k}x_k + W_{j+1,I}x_I + b_{j+1} \right]_+  \\
        &  + x_{j} - \left[ \sum_{k=1}^{n} W_{jk}x_k + W_{jI}x_I + b_{j} \right]_+.
    \end{align*}
   In the chamber $R_{[n]}$, all $y_i(x) >0$ for $i\in [n]$, so that the threshold nonlinearity $[\ ]_+$ can be omitted. Because $ W_{ii} = c$, $ W_{iI} = -1$, $b_i= \theta$ for all $i \in [n]$, and both nodes $j+1$ and $j$ have the same uniform in-degree $d_{in}$, we obtain
    \begin{align*}
        \frac{d z_j}{dt}  
         = & -x_{j+1} + \left(W_{j+1,j+1}\frac{x_E}{n} + d_{in}a \frac{x_E}{n}+\sum_{k=1}^{n} W_{j+1,k}f_k(z)-x_I + \theta \right)  \\
       & + x_{j} - \left(W_{jj}\frac{x_E}{n} + d_{in}a \frac{x_E}{n} + \sum_{k=1}^{n} W_{jk}f_k(z)-x_I + \theta \right)  \\
        = & -z_j + \sum_{k=1}^{n} (W_{j+1,k} - W_{jk})f_k(z).
    \end{align*}
    Thus, since $W_{j+1,k} - W_{jk}$ is constant and $f_k(z)$ is a function of $z$, $\dot z_j$ depends only on $z$ and it is independent of $(x_E,x_I)$. 
    Therefore, when $G$ has  uniform in-degree, the $z$-mode decouples completely from the mean variables. 
    
    Next, suppose $G$ has  uniform out-degree $d_{out}$. Summing over all $i\in [n]$ in the chamber $R_{[n]}$, we obtain
    \begin{align*}
        \frac{d x_E}{dt} = & \sum_{j=1}^{n} \frac{d x_{j}}{dt} = - x_E + (d_{out} a + c) x_E - n x_I + n \theta, \\
        \frac{d x_I}{dt}  =& - x_I + c x_E,
    \end{align*}
    which involves only $(x_E,x_I)$.  
    Hence, in this case, the mean mode decouples from the $z$-mode.
\end{proof}

\thmthree*
\begin{proof} 
    We analyze the two chambers $R_\emptyset$ and $R_{[n]}$ separately.
    In the chamber $R_\emptyset$, the $z$-mode is given by
    \begin{equation*}
         \frac{dz_i}{dt} = -z_i, \quad i=1,\dots,n-1,
    \end{equation*}
    while the mean mode is
    \begin{align*}
     \frac{dx_E}{dt} &= -x_E, \\
     \tau_I \frac{dx_I}{dt} &= -x_I + c x_E.
    \end{align*}
    Here the $z$-mode and mean mode are completely decoupled. The $z$-mode has a single fixed point $z^*=0$, which is always stable since all its eigenvalues equal $-1$.  
    
    On the other hand, in the chamber $R_{[n]}$, under the ordering $(1,2,\dots,n)$, the $z$-mode is given by
    \begin{align*}
    \frac{dz_1}{dt} &= (-1+c)\,z_1 - a\sum_{j=1}^{n-1} z_j, \\
     \frac{dz_i}{dt} &= (-1+c)\,z_i + a\,z_{i-1}, \qquad i=2,\dots,n-1. 
    \end{align*}
    These equations involve only $z=(z_1,\dots,z_{n-1})$, so the $z$-mode can be studied independently of the mean mode. 
    We rewrite the $z$-mode in matrix form: $\dot z = M_{n-1}z$ with
    \begin{align*}
    M_{n-1}
    =
    \begin{pmatrix}
    -1+c-a & -a  & \cdots & -a  & -a \\
    a  & -1+c & \cdots & 0   & 0 \\
    \vdots & \vdots & \ddots & \vdots   & \vdots \\
    0  & 0  & \cdots & -1+c & 0 \\
    0 & 0 & \cdots & a & -1+c  
    \end{pmatrix}
    \;=\;
    (c-1) I_{n-1} + a B,
    \end{align*}
    where $I_{n-1}$ is the $(n-1)\times(n-1)$ identity matrix and $B$ is defined by\footnote{Although the underlying graph has cyclic symmetry, the matrices $M_{n-1}$ and $B$ are not circulant. This is because we use the independent difference variables $z=(z_1,\dots,z_{n-1})$ and omit $z_n=x_1-x_n$, which breaks the cyclic symmetry at the level of this coordinate representation.}
    $$
    B_{i+1,i}=1\quad i=1,\dots,n-2, 
    \qquad
    B_{1,j}=-1\quad j=1,\dots,n-1,
    $$
    with all remaining entries equal to zero. Hence 
    the eigenvalue-eigenvector pairs of $B$ are $\bigl(\omega_k,\ v_k\bigr)$:
    \begin{equation*}
    \omega_k = e^{2\pi i k/n},\quad
    v_k = (\omega_k^{n-2}, \omega_k^{n-3},\dots,1)^\top,
    \end{equation*}
    for $k=1,\dots,n-1$.
    Therefore the eigenvalues of $M_{n-1}$ are 
    \begin{equation*}
        \lambda_j=c-1+a e^{\frac{2\pi j}{n}i}, \quad j = 1,\dots,n-1.
    \end{equation*}
     Thus, the $z$-mode is stable if and only if $\Re(\lambda_j)<0$ for all $j=1,\dots,n-1$, equivalently $c < 1 - a\cos\frac{2\pi}{n}$,
     and unstable when $c>1-a\cos\frac{2\pi}{n}$. 
    
    For the mean mode in the chamber $R_{[n]}$ we have
    \begin{align*}
     \frac{dx_E}{dt} &= (-1+c+a)\,x_E - n x_I + n\theta, \\
     \tau_I \frac{dx_I}{dt} &= -x_I + c x_E,
    \end{align*}
    with the Jacobian matrix 
    \begin{equation*}
    M_{EI}=\begin{pmatrix}
    -1+c+a & -n\\
    \tfrac{c}{\tau_I} & -\tfrac{1}{\tau_I}
    \end{pmatrix}.
    \end{equation*}
    Stability requires $\operatorname{tr}(M_{EI})<0$ and $\det(M_{EI})>0$.  
    The trace inequality gives
    $$
    a+c < 1 + \tfrac{1}{\tau_I}.
    $$
    Meanwhile, $\det(M_{EI})>0$ is equivalent to $c>\tfrac{a-1}{n-1}$, which is precisely the existence condition for the full-support fixed point.  
    Since we are already assuming the fixed point exists, this inequality is automatically satisfied.  
    Therefore the mean mode is stable when $ a+c < 1 + \tfrac{1}{\tau_I}$ and unstable when $ a+c > 1 + \tfrac{1}{\tau_I}$.
\end{proof}




\subsection{Example of $2$-path E-I TLN}
\label{sec:2-path}

Here we consider a slightly more complex structure: the $2$-path with $1 \to 2$ and $2 \not\to 1$. For this graph $G$, the dynamics of the corresponding E-I TLN are given by
\begin{subequations}
\begin{align}
           \frac{dx_1}{dt} &= -x_1 + \left[ c x_1 - x_I + \theta \right]_+, \\
           \frac{dx_2}{dt} &= -x_2 + \left[ c x_2 + ax_1 - x_I + \theta \right]_+, \\
    \tau_I \frac{dx_I}{dt} &= -x_I + \left[ c x_1 + cx_2         \right]_+. 
\end{align}
\end{subequations}
Now, each nonempty subset of $\{1,2\}$ yields a candidate e-support.  A direct computation shows that the network
admits at most three e-supports:
\begin{enumerate}
    \item $\esupp(x^*) = \{1\}$. Since $\big|-I+W|_{\{1,I\}}\big|=1$, the corresponding fixed point is unique (if it exists) and given by  $(x^*_1,x^*_2,x^*_I) = (\theta,0,c\theta)$.
    This fixed point exists when $y_1(x^*)>0$ and $ y_2(x^*) \leq 0$, i.e.,\ $c \geq a+1$.
    \item $\esupp(x^*) = \{2\}$. Similar to $\{1\}$,  $\big|-I+W|_{\{2,I\}}\big|=1$, and the unique fixed point is
    $(x^*_1,x^*_2,x^*_I) = (0,\theta,c\theta)$.
    It exists when $y_1(x^*) \leq 0$ and $ y_2(x^*)>0$, i.e.,\ $c \geq 1$.
    \item $\esupp(x^*) = \{1,2\}$. By calculating $\big|-I+W|_{\{1,2,I\}}\big|=c^2-ac-1$, we find the E-I TLN in the chamber $R_{12}$ is degenerate when $c^2-ac-1 = 0$ and nondegenerate otherwise. 

    When $c^2-ac-1 = 0$, $-I+W|_{\{1,2,I\}}$ is degenerate and the fixed point condition requires
    $$
    x_1+cx_2 = \theta, \qquad x_1+cx_2 = c\theta.
    $$
    These two equations hold only when $c=1$, which leads to $a=0$ since $c^2-ac-1 = 0$. However, we only consider the positive $(a,c)$ parameters. Thus, there are no fixed points here. 
    
    When $c^2-ac-1 \neq 0$, the unique fixed point is
    $$
    (x^*_1,x^*_2,x^*_I) 
    = \Biggl(\frac{1-c}{1+ac-c^2}\,\theta,\; \frac{a+1-c}{1+ac-c^2}\,\theta,\;  
      \frac{ac+2c-2c^2}{1+ac-c^2}\,\theta\Biggr),
    $$
    which exists when $y_1(x^*)>0$ and $ y_2(x^*)> 0$. This occurs precisely when $c>a+1$ or $c<1$.
\end{enumerate}
Based on the existence condition of these fixed points, it is natural to divide the $(a,c)$ parameter space ($a>0, c>0$) into three regions: the strong inhibition regime ($c>a+1$), the moderate inhibition regime ($1<c<a+1$), the weak inhibition regime ($c<1$) and two boundaries ($c=1$ and $c=a+1$). In the strong inhibition regime, all three supports survive. In the moderate inhibition regime, $\{2,I\}$ is the unique support, while in the weak inhibition regime, only $\{1,2,I\}$ can support a fixed point.

Next, we analyze stability.  
The fixed points supported on $\{1,I\}$ or $\{2,I\}$ have the same reduced Jacobian as in the singleton system, since the induced subgraphs $G|_{\{1\}}$ and $G|_{\{2\}}$ are themselves singletons.  
Thus, each is stable when $c<1+\tfrac{1}{\tau_I}$ and unstable when $c>1+\tfrac{1}{\tau_I}$.  
For the full-support fixed point $\{1,2,I\}$, the Jacobian is
$$
J=\begin{pmatrix}
c-1 & 0 & -1 \\
a & c-1 & -1 \\
\tfrac{c}{\tau_I} & \tfrac{c}{\tau_I} & -\tfrac{1}{\tau_I}
\end{pmatrix}.
$$
When $c>a+1$, this fixed point is always unstable.  
For $0<c<1$, the Routh-Hurwitz conditions show that it is stable exactly when
$$
0<a<a_{\max}(c,\tau_I),\qquad
a_{\max}(c,\tau_I)=\frac{c^2 -1 +\tau_I\bigl(-2c+2+\tfrac{1}{\tau_I}\bigr)
           \bigl((c-1)^2 + \tfrac{2}{\tau_I}\bigr)}{c},
$$
and unstable when $a>a_{\max}(c,\tau_I)$. Finally, combining the existence and stability conditions, the positive $(a,c)$ parameter space is divided into 6 distinct regions. Across these regions, numerical simulations reveal seven qualitatively different dynamical states, denoted by seven points (P1–P7), as illustrated in Figure~\ref{fig:2path}B–C.

\begin{figure}
    \centering
    \includegraphics[width=1\linewidth]{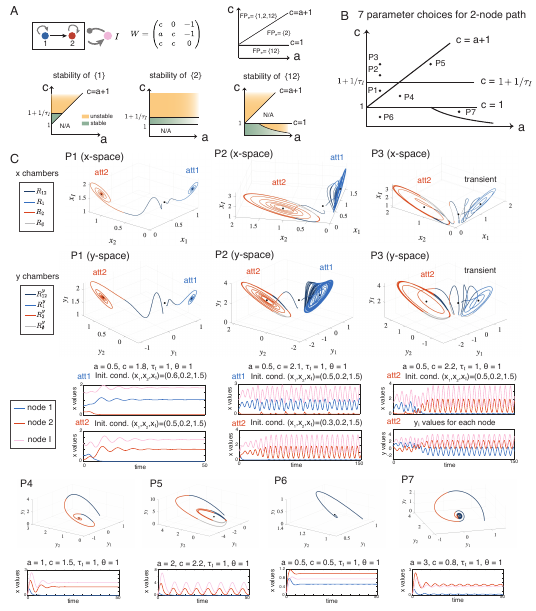}
   \caption{
    Seven parameter choices of the E-I TLN whose underlying graph is a 2-node path.
    (A) Network architecture and connectivity matrix of the 2-path E-I TLN, together with the existence and stability of fixed points with different supports.
    (B) Six parameter regions corresponding to the seven states (denoted by seven points), determined by the existence and stability conditions of the fixed points. 
    Points~P2 and~P3 occupy the same parameter region but exhibit distinct dynamical behaviors.
    (C) Dynamical attractors associated with each of the seven points; black dots correspond to fixed points. 
    The initial conditions for P1-P3 are indicated above the rate curves, whereas those for P4-P7 are vectors with entries sampled randomly from $(0,0.1)$. 
    Trajectories pass through different chambers, each identified by a unique color. The colors in the activity plots correspond to the node colors shown in panel~A.
}
    \label{fig:2path}
\end{figure}

In the strong inhibition regime, all three possible fixed points appear at the parameter points P1, P2, and P3. Among these, the fixed point on $\{1,2,I\}$ is always unstable. The qualitative differences between these states arise as $c$ crosses the stability threshold $c = 1 + 1/\tau_I$. When $c < 1 + 1/\tau_I$ (P1), the singleton fixed points e-supported on $\{1\}$ and $\{2\}$ are both stable.
When $c > 1 + 1/\tau_I$ (P2), these two fixed points lose stability and each gives rise to a dynamical attractor.
The attractor emerging from $\{1\}$ resembles a M\"obius-like chaotic attractor and corresponds to an E–I oscillation between a single excitatory node and the inhibitory node. Additional examples of this behavior are shown in Figure~\ref{fig:2path_chaos}.
Within the same parameter region ($c > 1 + 1/\tau_I$), a second dynamical pattern (shown by P3) is also observed: the chaotic attractor around $\{1,I\}$ disappears, yet trajectories may still linger near $\{1,I\}$ for a while before ultimately converging to the attractor around $\{2,I\}$.

For parameter choices P1, P2, and P3, we display the corresponding attractors in both $x$-space and $y$-space. Since the variables satisfy $x_i \ge 0$, the dynamics are confined to the nonnegative orthant. As a result, certain attractors—such as the chaotic attractor associated with $\{1,I\}$—partially evolve along boundary hyperplanes, which can obscure their full geometric structure. In contrast, the transformed variables $y$ are not constrained to be nonnegative, allowing the attractors to unfold fully in $y$-space. For this reason, all attractors are presented in $y$-space in the following figures.

In the moderate inhibition regime, the only possible support is $\{2,I\}$, which is in agreement with the classical CTLN results. Depending on the stability of this unique fixed point, the E–I system exhibits two qualitatively distinct behaviors: a stable fixed point (P4) or an E–I oscillation (P5) around $\{2,I\}$.
Notably, the periodic attractor does not remain confined to the chamber $R_2$ and $R_\emptyset$; instead, it repeatedly traverses $R_{2}$, $R_{12}$, and $R_{\emptyset}$. This shows that even though the support does not contain node $1$, it does not remain silent throughout this oscillatory behavior.

In the weak inhibition regime, $\{1,2,I\}$ is the unique support. Hence, the system displays two distinct behaviors depending on the stability of the fixed point supported on $\{1,2,I\}$.
For the parameter point P6, this fixed point is stable, and trajectories converge to an equilibrium where node $1$ settles at a larger activity value than node $2$. In contrast, at P7, the fixed point becomes unstable, giving rise to a small-amplitude E-I oscillation supported on $\{1,2,I\}$.
Unlike the singleton case where oscillations require sufficiently strong inhibition ($c>1+1/\tau_I$), here oscillations occur already for $c<1$.
This illustrates that, relative to the singleton network, the presence of a nontrivial graph structure can substantially enlarge the parameter region that supports E–I oscillations.

We now turn to the two boundary cases $c=a+1$ and $c=1$.  
An interesting observation is that at $c = a+1$, $\FPe(G,a,c) = \{\{1\},\,\{2\}\}$, whose cardinality is even!
At first glance, this appears to contradict the parity theorem in~\cite{morrison2024diversity}, which states that any nondegenerate TLN must have an odd number of fixed points.  
The resolution is that the E-I TLN at $c=a+1$ is actually degenerate: condition~(2) in the definition of nondegeneracy fails, and thus the odd-number theorem does not apply.  
For this reason, throughout the paper we exclude the boundary $c=a+1$ from our analysis. A similar issue arises at $c = 1$.  When $G|_\sigma$ is an independent set of size at least two, we obtain  
$\det\left(I - W_{\sigma \cup \{I\}}\right) = 0$ for $c = 1$, so the network becomes degenerate as well.  
Thus, we will also exclude the boundary $c=1$ from consideration.

\subsection{Supplementary figures}
\begin{figure}[H]
    \centering
    \includegraphics[width=.9\linewidth]{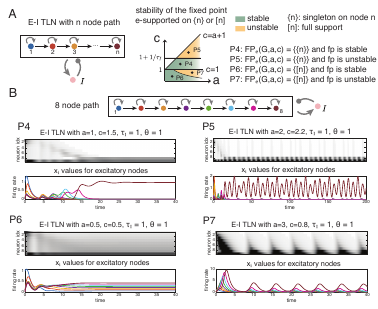}
    \caption{
    Emergent dynamics in the moderate and weak inhibition regimes for an E-I TLN whose underlying graph is an $n$-path.
    (A) Network architecture and stability conditions for fixed points e-supported on $\{n\}$ and $[n]$, where $[n]=\{1,2,\ldots,n\}$.
    (B) Four dynamical states (shown by P4-P7) of the $n$-path E-I TLN, illustrated for the $8$-path case. In each state, the top panel shows firing rates over time in grayscale, and the bottom panel shows node-colored activity. The initial conditions are vectors whose entries are sampled randomly from $(0,0.1)$.
    }
    \label{fig:npath_other_inhibition}
\end{figure}

\begin{figure}[H]
    \centering
    \includegraphics[width=0.9\linewidth]{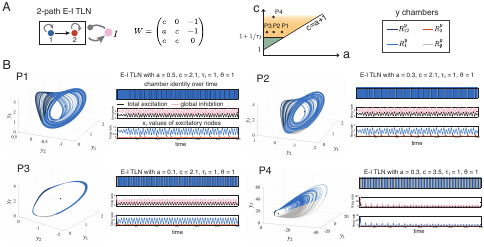}
    \caption{Different chaotic attractors around the singleton fixed point $\{1\}$ in an E-I TLN whose underlying graph is a 2-path.
    (A) Connectivity matrix for the 2-path E–I TLN and the $(a,c)$ parameter choices.
    (B) Möbius-like chaotic oscillations under different parameter choices. All simulations use the same initial condition $(1,0.1,2)$, chosen to activate node 1.}
    \label{fig:2path_chaos}
\end{figure}

\begin{figure}[H]
    \centering
    \includegraphics[width=0.9\linewidth]{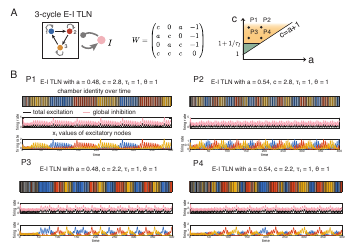}
    \caption{Different sequential chaotic oscillations in an E-I TLN whose underlying graph is a 3-cycle.
    (A) Connectivity matrix of the 3-cycle E–I TLN and the $(a,c)$ parameter choices.
    (B) Sequential chaotic oscillations exhibiting similar qualitative behaviors but different patterns. The initial conditions are vectors whose entries are sampled randomly from $(0,0.1)$.}
    \label{fig:3cycle_sco}
\end{figure}

\end{document}